\documentclass{article}
\usepackage{arxiv}
\usepackage[utf8]{inputenc} 
\usepackage[T1]{fontenc}    
\usepackage{hyperref}       
\usepackage{url}            
\usepackage{booktabs}       
\usepackage{amsfonts}       
\usepackage{nicefrac}       
\usepackage{microtype}      
\usepackage{lipsum}
\usepackage{tikz}
\usetikzlibrary{arrows,shapes,quotes}
\usetikzlibrary{patterns,decorations.pathreplacing}
\usepackage{amssymb, graphicx}
\usepackage{hyperref,lscape}
\usepackage{blkarray}
\usepackage{color}
\usepackage{amsmath}
\usepackage{mathtools}
\usepackage{mathrsfs}
\usepackage{chngcntr}
\usepackage{apptools}
\usepackage{blkarray}
\usepackage{csquotes}
\usepackage{bm}

\usepackage{parskip}
\usepackage{multicol}
\usepackage{amssymb}
\usepackage{cite}
\usepackage{amsthm}
\usepackage{makeidx}
\usepackage{graphicx}
\usepackage{mathrsfs,xcolor}
\usepackage{enumerate}
\usepackage{blkarray}
\usepackage{bm}
\usepackage{setspace}
\usepackage{algorithm,algorithmic,refcount}
\usepackage[left,pagewise, mathlines]{lineno}
\usepackage{tikz-cd} 
\usepackage{lineno}
\usepackage{comment}
\usepackage{longtable}
\usepackage{tabularx}

\DeclareMathOperator{\Ima}{Im}
\DeclareMathOperator\supp{supp}
\DeclareMathOperator\spa{span}
\DeclareMathOperator{\sgn}{sgn}

\theoremstyle{plain}
\newtheorem{theorem}{Theorem}
\newtheorem{lemma}{Lemma}
\newtheorem{proposition}{Proposition}
\newtheorem{corollary}{Corollary}
\newtheorem{procedure}{Procedure}

\theoremstyle{definition}
\newtheorem{definition}{Definition}
\newtheorem{example}{Example}
\newtheorem{remark}{Remark}
\newtheorem{erratum}{Erratum}

\title{Network transformation-based analysis of biochemical systems}

\author{Dylan Antonio SJ.~Talabis \\
  Institute of Mathematical Sciences and Physics \\
  University of the Philippines \\
  Los Banos, Laguna 4031, Philippines \\
  \texttt{dstalabis1@up.edu.ph} \\
  \And
Eduardo R.~Mendoza\thanks{Max Planck Institute of Biochemistry, Martinsried near Munich, Germany} \thanks{Faculty of Physics, Ludwig Maximilian University, Munich 80539 Germany} \\
  Mathematics and Statistics Department \\
  Center for Natural Sciences and Environmental Research \\
  De La Salle University \\
  Manila 0922, Philippines \\
  \texttt{eduardo.mendoza@dlsu.edu.ph} \\
  }

\begin{document}
\maketitle

\begin{abstract}
A dynamical system obtains a wide variety of kinetic realizations, which is advantageous for the analysis of biochemical systems. A reaction network, derived from a dynamical system, may or may not possess some properties needed for a thorough analysis. We improve and extend the work of  M. Johnston \cite{JOHN2014} and Hong et al. \cite{HONG2023} on network translations to network transformations, where the network is modified while preserving the dynamical system. These transformations can shrink, extend, or retain the stoichiometric subspace. Here, we show that positive dependent network can be translated to a weakly reversible network. Using the kinetic realizations of (1) calcium signaling in the olfactory system and (2) metabolic insulin signaling, we demonstrate the benefits of transformed systems with positive deficiency for analyzing biochemical systems. Furthermore, we present an algorithm for a network transformation of a weakly reversible non-complex factorizable kinetic (NFK) system to a weakly reversible complex factorizable kinetic (CFK) system, thereby enhancing the Subspace Coincidence Theorem for NFK systems of Nazareno et al. \cite{NAZA2019}.  Finally, using the transformed kinetic realization of monolignol biosynthesis in \textit{Populus xylem}, we study the structural and kinetic properties of transformed systems,including the invariance of concordance and variation of injectivity and mono-/multi-stationarity under network transformation.
\end{abstract}

\baselineskip=0.30in

\section{Introduction}\label{sec1}
Mathematical models of biochemical systems are mostly dynamical systems, i.e.,  given as solutions of a set of ordinary differential equations (ODE). A set of molecular entities and their interactions, typically illustrated as a “biochemical map”, underlies such equations. This relationship can advance  the understanding of the system's properties when used to formally construct a kinetic realization of the system. We recall the concept of a kinetic system as follows:
\begin{itemize}
    \item A reaction network $\mathscr{N}$ is  a directed graph whose vertices $y$ (called complexes) are non-negative vectors $y$ in $\mathbb{R}^m$ and whose edges (called reactions) are denoted by $y \rightarrow y'$ . The units in $\mathbb{R}^m$  are called species and in our case correspond to the state variables of the ODE system.  A CRN is typically denoted as a triple $\mathscr{N} = (\mathscr{S}, \mathscr{C}, \mathscr{R})$, where $\mathscr{S}$ is the set of $m$ species, $\mathscr{C}$  the set of $n$ complexes and $\mathscr{R}$ the set of $r$ reactions.
    \item A kinetics $K$ on a CRN $\mathscr{N}$ assigns to each reaction $q: y \rightarrow y'$ a rate function $K_q: \Omega \rightarrow \mathbb{R}_{\geq}$, where $\Omega$ is a subset of $\mathbb{R}^m_{\geq}$ containing $\mathbb{R}^m_{>}$, and $\supp x$ contains $\supp y$ $\Rightarrow$ $K_q(x) > 0$.
    \item A kinetic system is a pair $(\mathscr{N}, K)$ with $\mathscr{N}$ a CRN and $K$ a kinetics on $\mathscr{N}$. Each reaction $q: y \rightarrow y'$ defines a reaction vector $y'- y$ in $\mathbb{R}^m$, and the CRN's stoichiometric matrix $N$ is the $m \times r$ matrix whose columns are the $r$ reaction vectors. The image on $N$ (in $\mathbb{R}^m$) is called the CRN's stoichiometric subspace $S$. If $K(x)$ denotes the $r$-vector of rates of $x$, then the function $f(x)= NK(x)$, called the system's species formation rate function, defines the system's vector field, i.e., the RHS of its system of ODEs.
    \item A natural equivalence relation between kinetic systems is dynamical equivalence. $(\mathscr{N}, K)$ and $(\mathscr{N}', K')$ are dynamically equivalent if (i) $\mathscr{N}$ and $\mathscr{N}'$ have the same set of species $S = S'$, (ii) $K$ and $K'$ have the same definition domain $\Omega = \Omega'$ and (iii) $NK (x) = N'K'(x)$ for all $x \in \Omega$. The last condition means that their ODE systems are identical.
\end{itemize}
We can now formally define a kinetic realization of a dynamical system: a kinetic system $(\mathscr{N}, K)$ is a kinetic realization of a dynamical system given by $dx_i/dt = f_i(x)$, for $i = 1,\ldots,m $ if (i) the species of $\mathscr{N}$ coincide with the $m$ state variables $x_i$ and (ii) $(NK (x))i = dx_i/dt$ for all $x \in \Omega$ and $i=1,\ldots,m$. \\

A dynamical system obtains a wide variety of kinetic realizations, since any dynamically equivalent of a kinetic realization is also a kinetic realization. This diversity is advantageous for the analysis of the biochemical system, since the CRN immediately derived from the biochemical map and equations may not possess the appropriate properties for a thorough analysis. However, caution is needed since many system properties are not necessarily preserved under dynamic equivalence. One of the goals of this paper is to provide new approaches for such analysis. \\

Early research in identifying dynamic equivalents of a given kinetic system was based on methods from Mixed Integer Linear Programming (MILP). Many of the results were also extended to the more general relation of linear conjugacy. An excellent overview of these studies is provided in the Introduction of the paper by M. Johnston, D. Siegel and G. Szederkenyi \cite{JOSS2013}. \\

Here, our approach is based on the network translation method introduced by M. Johnston in 2014 for mass action systems \cite{JOHN2014}. The network translate of a mass action system is a generalized mass action system in the sense of S. Müller and G. Regensburger \cite{MURE2012}. The novel approach provided a geometric interpretation of the “toric steady states” of the mass action system previously discovered by Craciun et al. in 2007 as complex balanced equilibria of the generalized mass action system. In the biochemical context, Talabis et al. \cite{TAME2018} showed that generalized mass action systems corresponded to PL-RDK systems, i.e. power law kinetic systems where all branching reactions of any reactant had identical kinetic order vectors. \\

Recently, Hong et al. \cite{HONG2023} extended the concept of network translation to all kinetic systems and derived important properties of the method in this generality. Their contributions include the following:

\begin{itemize}
    \item Demonstration that any network translation of a kinetic system is a sequence of three network operations: shifting, dividing and merging. (See Section \ref{Sec:2} for precise definitions),
    \item Derivation of a necessary and sufficient condition for a kinetic system to have weakly reversible network translate, and
    \item Design and MATLAB implementation of TOWARDZ, an algorithm for finding weakly reversible, deficiency zero network translates of mass action systems (both deterministic and stochastic).
\end{itemize}
    
This paper aims to improve and extend the contributions of Hong et al. in several aspects. \\

We \textbf{improve} their results by
\begin{itemize}
    \item Showing the necessary and sufficient condition is equivalent to the well-known property of positive dependence.
    \item Broaden the analysis perspective to weakly reversible network translates with positive deficiency in two ways: 
    \begin{itemize}
        \item First, we provide two examples of beneficial analysis of biochemical systems (olfactory calcium signaling and insulin signaling in adipocytes) despite their translates' positive deficiencies.
        \item Second, for any weakly reversible NFK system (of arbitrary deficiency), we construct an algorithm for finding a weakly reversible CFK network translate. We then use the translate to improve results of Nazareno et al. \cite{NAZA2019}, including their Subspace Coincidence Theorem for NFK systems.
    \end{itemize}    
\end{itemize}
    
We \textbf{extend} the work of Hong et al. by introducing network transformation as a generalization of network translation and defining 5 additional network operations to generate such transformations. In contrast to the operations for network translation, these new operations do not leave the set of reaction vectors invariant, but contract, expand or retain the stoichiometric subspace. Hence, a network transformation is a dynamic equivalence that contracts, expands or leaves the stoichiometric subspace $S$ invariant. A network translation is an $S$-invariant network transformation with the special property that it leaves the subspace's canonical set of generators unchanged. \\

We then study the variation of system properties under different classes of network transformations, with a focus on concordance, a network property inducing stability (or even “dull behavior” \cite{SF2012}) in the system. In addition to the calcium and insulin signaling systems previously introduced, we illustrate the analysis with the PL-NDK realization of a model of biosynthesis in   \textit{Populus xylem}. To our knowledge, this is the first example of a concordant non-mass action system in the literature. We also discuss applications of the analysis to results in Nazareno et al. \cite{NAZA2019}, Fortun and Mendoza \cite{FOME2020} and Hernandez et al. \cite{HELM2024}. \\

The paper is organized as follows: in Section \ref{Sec:2}, we review the work of Hong et al. and present our clarification of some of their results. Section \ref{Sec:3} introduces the olfactory calcium signaling and adipocyte insulin signaling systems and analyzes their positive deficiency network translates. In Section \ref{Sec:4}, we discuss the algorithm for constructing a weakly reversible CFK translate for a weakly reversible NFK system (regardless of its deficiency value). Section \ref{Sec:5} applies the algorithm to improve results in Nazareno et al., including the Subspace Coincidence Theorem for NFK systems. In Sections \ref{Sec:6} and \ref{Sec:7}, we extend the study to network transformations, introduce the additional network operations and illustrate the variation of system properties with the model of \textit{Populus xylem}. Section \ref{Sec:7} focuses on network concordance, including a sufficient condition for invariance and on improvements of results in previous publications. We conclude with a Summary and Outlook in Section \ref{Sec:8}. An Appendix collects the fundamentals on reaction networks and kinetics useful for the paper.

\section{Network translation of kinetic systems}
\label{Sec:2}
In this Section, we review the main results of Hong et al. \cite{HONG2023} and clarify aspects of their contribution with several new propositions.

\subsection{The concept of network translation}
M. Johnston pioneered the study of network translation of mass action systems in 2014 \cite{JOHN2014}. His main motivation was to find an interpretation of the “toric steady states” of the mass action systems discovered by Craciun et al. in 2009 \cite{CRACIUN2009}. His approach led him to consider power law kinetic systems, which he formulated in terms of the generalize mass action systems introduced by Müller and Regensburger in 2012 \cite{MURE2012}. In the biochemical context, Talabis et al. demonstrated that these systems are the PL-RDK systems, where all branching reactions of any reactant share identical kinetic order vectors \cite{TAME2018}. Recently, Hong et al. \cite{HONG2023} extended the concept of network translation to all kinetic systems and derived important properties of the method. We review these contributions and provide clarifications for their interpretation. \\

In Hong et al., the fact that a network translation is a dynamic equivalence that leaves the stoichiometric subspace invariant follows implicitly from its generation through 3 network operations (cf. discussion in the next section). Here, we derive the property directly from the definition of network translation.
 
\begin{definition}
    For a kinetic system $(\mathscr{N},K)$, we call $(\tilde{\mathscr{N}},\tilde{K})$ a \textbf{translation} of $(\mathscr{N},K)$ if 

    \begin{equation}
    \label{hong}
        \sum_{r:y'-y=\zeta} K_r (\textbf{x}) = \sum_{\tilde{r}:z'-z=\zeta} \tilde{K}_{\tilde{r}} (\textbf{x})
    \end{equation}  
    for any $\zeta \in \mathbb{Z}^m$ and $x \in \mathbb{R}^m_{\geq 0}$.
\end{definition}

We state a necessary condition for a kinetic system $(\tilde{\mathscr{N}},\tilde{K})$ to be a network translate of a given $(\mathscr{N},K)$:
 
\begin{proposition}
    \label{prop:reacvec}
    Let $\sigma: \mathscr{R} \rightarrow \mathbb{R}^m$ maps a reaction to its reaction vector. If $(\tilde{\mathscr{N}},\tilde{K})$ be a network translate of $(\mathscr{N},K)$, then $\sigma (\mathscr{R}) = \sigma (\tilde{\mathscr{R}})$, i.e., their sets of reaction vectors coincide.
\end{proposition}

\begin{proof}
    We observe that the sums on both sides of Equation \ref{hong} collect the values of the kinetics of a reaction vector which may be the same for several reactions. The sums are parametrized by $\zeta \in \mathbb{Z}^m$, with the sums taken over all reactions in $\mathscr{N}$(LHS) and $\tilde{\mathscr{N}}$(RHS) whose reaction vector $= \zeta$. If the set of reaction vectors of the networks are different, say, for a reaction vector of $\mathscr{N}$ not a reaction vector of $\tilde{\mathscr{N}}$, for any positive $x$ in $\mathbb{R}^m$, LHS is positive, while the  RHS is zero (the index set of the sum is empty). Hence, the networks cannot be translates. This similarly applies to a reaction vector in $\tilde{\mathscr{N}}$ but not in $\mathscr{N}$.
\end{proof}

\begin{corollary}
    Let $S$ and $\tilde{S}$ be the stoichiometric subspaces of the networks $\mathscr{N}$ and $\mathscr{\tilde{N}}$.   If $(\mathscr{N},K)$ and $(\mathscr{\tilde{N}},\tilde{K})$ are network translations (or translates), then they are dynamically equivalent with $S = \tilde{S}$.
\end{corollary}

\begin{proof}
    
    Summing the LHS of Equation \ref{hong} over $\sigma(\mathscr{R})$ gives the vector field of $(\mathscr{N},K)$. Analogously, summing the RHS over $\sigma(\tilde{\mathscr{R}})$ results in that of $(\mathscr{\tilde{N}},\tilde{K})$. Since $\sigma(\mathscr{R}) = \sigma(\tilde{\mathscr{R}})$, we obtain the dynamical equivalence and stoichiometric subspace coincidence.
\end{proof}

\subsection{Network translation and network operations}
Following from Proposition \ref{prop:reacvec}, only operations which preserve the reaction vector can be considered as network translation. These are shifting, dividing and merging. 
 
\begin{definition}[Shifting, Dividing and Merging]
Let $q: y \xrightarrow{r} y'$ be a reaction in a kinetic system $(\mathscr{N}, K)$. 
\begin{enumerate}

    \item \textbf{Shifting} the reaction results to  $q':y + z \xrightarrow{r} y' + z$ with kinetics $K_{q'}= K_{q}$ where $z = z_1 X_1 + ... + z_m X_m$ and $z_i$'s $\in \mathbb{Z}$.
    
    \item \textbf{Dividing} the reaction results to the pair of reactions $q': y + a \xrightarrow{r_1} y' + a$ and $q'': y + b \xrightarrow{r_2} y' + b$ with kinetics $K_{q'}=r_1 I_{K,1}$ and $K_{q''}=r_2 I_{K,1}$, respectively. Here, $K_{q}=K_{q'} + K_{q''}$, $a = a_1 X_1 + ... + a_m X_m$, $b = b_1 X_1 + ... + b_m X_m$ and $a_i$'s, $b_i$'s $\in \mathbb{Z}$.

    \item Let $q_1: y \xrightarrow{r_1} y'$ and $q_2: z \xrightarrow{r_2} z'$ be reactions with the same reaction vector and kinetics $K_{1}$ and $K_{2}$, respectively. \textbf{Merging} the reactions results to the reaction $q': x \xrightarrow{r} x'$ with kinetics $K_{q'}=K_1 + K_2$ where $x'-x = y'-y$.
         
\end{enumerate}
\end{definition}
In the example below, we shift the reaction $R_1: Y \rightarrow X $ by adding Y to both reactant and product complexes yielding the reaction $2Y \rightarrow X+2Y$ and making the network weakly reversible. In the succeeding examples for dividing and merging, $R_3$ results in two reactions $Y \rightarrow X$ and $2Y \rightarrow X +X$ and $R_1^*$ is the merge of $R_1$ and $R_2$. See Table \ref{tabexxi} for complete details. Here, each reaction $R_i$ is associated to kinetics $K_i$. Note that the kinetics $K_i$ can be decomposed as $K_i = r_i I_{K,i}$ where $r_i \in \mathbb{R}_+$.

\begin{table}[h]
\caption{Examples for the three (reaction) network operations}
\label{tabexxi}
\begin{tabular*}{\textwidth}{@{\extracolsep\fill} p{0.11\linewidth} llll}
\toprule%
& \multicolumn{2}{@{}c@{}}{Original Kinetic System} & \multicolumn{2}{@{}c@{}}{WR Dynamic Equivalent System} \\\cmidrule{2-3}\cmidrule{4-5}%
Operation & Network & Kinetics & Network & Kinetics \\

\midrule

Shifting &
$\begin{array}{l}
R_1: Y \rightarrow X\\
R_2: X+Y \rightarrow 2Y
\end{array}$
 & $\begin{array}{l}
K_1 =r_1 I_{K,1}  \\
K_2 =r_2 I_{K,2}
\end{array}$ & $\begin{array}{l}
R_1^*: 2Y \rightarrow X+Y\\
R_2: X+Y \rightarrow 2Y
\end{array}$ & $\begin{array}{l}
K_1 =r_1 I_{K,1}  \\
K_2 =r_2 I_{K,2}
\end{array}$ \\

\midrule

Dividing &
$\begin{array}{l}
R_1: X+Y \rightarrow 2Y \\
R_2: X \rightarrow Y \\
R_3: Y \rightarrow X
\end{array}$
 & $\begin{array}{l}
K_1 =r_1 I_{K,1}  \\
K_2 =r_2 I_{K,2} \\
K_3 = r_3 I_{K,3} \\
+ r_4 I_{K,4}  \\
\end{array}$ & $\begin{array}{l}
R_1: X+Y \rightarrow 2Y \\
R_2: X \rightarrow Y \\
R_3: Y \rightarrow X \\
R_4^*: 2Y \rightarrow X+Y \\
\end{array}$
 & $\begin{array}{l}
K_1 =r_1 I_{K,1}  \\
K_2 =r_2 I_{K,2} \\
K_3 = r_3 I_{K,3}  \\ 
K_4 = r_4 I_{K,4} \\
\end{array}$ \\

\midrule

Merging & $\begin{array}{l}
R_1: X+Y \rightarrow 2Y \\
R_2: X \rightarrow Y \\
R_3: Y \rightarrow X \\
\end{array}$
 & $\begin{array}{l}
K_1 =r_1 I_{K,1}  \\
K_2 =r_2 I_{K,2} \\
K_3 =r_3 I_{K,3}  \\ 
\end{array}$ & $\begin{array}{l}
R_1^*: X \rightarrow Y \\
R_3: Y \rightarrow X \\
\end{array}$
 & $\begin{array}{l}
K_1 =r_1 I_{K,1} \\
+ r_2 I_{K,2} \\
K_3 =r_3 I_{K,3}  \\ 
\end{array}$ \\
\bottomrule
\end{tabular*}
\end{table}

A very important result of Hong et al. \cite{HONG2023} is the following:
 
\begin{proposition}
    For two kinetic systems $(\mathscr{N},K)$ and $(\tilde{\mathscr{N}},\tilde{K})$, if (\ref{hong}) holds then $(\mathscr{N},K)$ can be translated to $(\tilde{\mathscr{N}},\tilde{K})$ by shifting, dividing, and merging reactions. Conversely, if $(\tilde{\mathscr{N}},\tilde{K})$ is obtained by shifting, dividing, or merging reactions in $(\mathscr{N},K)$, then Equation \ref{hong} holds.
\end{proposition}
The operations presented in this section are already studied in previous works, specifically their application in mass action. Please refer to Table \ref{app} for an overview of the applications of these operations as documented in earlier studies.

\begin{table}[h!]
\centering
\caption{Applications of network translating operations}\label{app}%
\begin{tabular}{lc}
\hline
Operations & Applications\\
\hline
Shifting & dynamic equivalence \cite{JOHN2014}, parametrization of steady states \cite{TONJOHN2018, JOHNMUPA2018}    \\
Dividing & ACR analysis \cite{MESH2022}  \\
Merging & derivation of stationary distributions \cite{HKK2021}  \\
\hline
\hline
\end{tabular}
\end{table}

\subsection{A necessary and sufficient condition for weakly reversible network translates}

Identifying a weakly reversible realization of a kinetic system is often very useful for analysis of the original system. Here, we show that the classical concept of positive dependence of a network is a necessary and sufficient condition for weakly reversible network translates. The property is usually defined as a special type of linear dependence of the network's reaction vectors:

\begin{definition}
    The reaction vectors of $\mathscr{N}$ are \textbf{positively dependent} if for each reaction $q: y \rightarrow y'$ there exists a positive number $k_q$ such that
    $$\sum_{q: y \rightarrow y' \in \mathscr{R}} k_q (y'-y) = 0.$$
\end{definition}

The following, more geometric characterization, does not seem to have been documented in the reaction network literature:
 
\begin{proposition}
\label{prop:posKerN*}
    $\mathscr{N}$ is positively dependent network $\Longleftrightarrow$ $\ker N$ contains a positive vector.
\end{proposition}

\begin{proof}
    A positive vector in $\mathbb{R}^\mathscr{R}$ can be written as $a=\sum_q a_q \omega_q$, where $\omega_q$ is the characteristic function of the reaction $q:y_q \rightarrow y_q'$ and $a_q > 0$. Hence, $I_a (a) = \sum_q a_q (\omega_{y'q} - \omega_{yq})$ where the $\omega$'s are the characteristic functions of the product and reactant complexes of $q$ in $\mathbb{R}^\mathscr{C}$. Applying the map $Y$ to both sides, we obtain $N(a) = \sum_q a_q(y_q' - y_q)$, which implies the claimed equivalence.
\end{proof}

We then obtain the following Proposition from \cite{FEIN1987} as a corollary, which directly identifies an important class of positively dependent reaction networks:
 
\begin{corollary}
The reaction vectors of a weakly reversible network are positively dependent.
\end{corollary}
\begin{proof}
    It is well known that a network is weakly reversible if and only if $\ker (I_a)$ contains a positive vector $z$. This vector is also contained in $\ker N$, since $N = YI_a$.
\end{proof}

We recall a further concept from \cite{FEIN1987}:

\begin{definition}
    By the \textbf{stoichiometric cone} for a network, denoted by $S_+$, we shall mean the set of vectos in $\mathbb{R}^m$ defined as follows: $\gamma \in \mathbb{R}^m$ is a member of the stoichiometric cone if and only if there exists a set of nonnegative numbers $\left\{ \alpha_{q} \right\}_{q \in \mathscr{R}}$ satisfying


$$\gamma = \sum_{q: y \rightarrow y' \in \mathscr{R}} \alpha_q (y'-y).$$

Clearly, the stoichiometric cone for a network is contained in and might be smaller than the stiochiometric subspace for the network.
\end{definition}

We now combine results from Feinberg \cite{FEIN1987}, Hong et al. \cite{HONG2023} and Proposition \ref{prop:posKerN*} to a comprehensive characterization of positive dependence in the following Theorem:

\begin{theorem}
\label{t13*}
    Let $(\mathscr{N},K)$ be a kinetic system, with $N$ and $S$, the stoichiometric matrix and subspace of $\mathscr{N}$, respectively.
    Then the following statements are equivalent:
    \begin{enumerate}
        \item[i)] $\mathscr{N}$ is positively dependent.
        \item[ii)] $\ker N$ contains a positive vector.
        \item[iii)] $S_+ = S$, where $S_+$ is the stoichiometric cone of $\mathscr{N}$.
        \item[iv)] There are rate constants such that $E_+(\mathscr{N},K) \neq \emptyset$, i.e., the system has a positive equilibrium.
        \item[v)] $(\mathscr{N},K)$ has a weakly reversible network translate.    
    \end{enumerate}
\end{theorem}

For proofs of the equivalences between i) and iii) – v), see \cite{FEIN1987} and \cite{HONG2023}.
\begin{remark}
    The previous Theorem is a striking example of the interplay between structural properties (statements i– iii) and kinetic properties (statements iv – v) of a kinetic system.
\end{remark} 
\begin{corollary}
    Let $\mathscr{N}=\mathscr{N}_1 \cup \cdots \cup \mathscr{N}_k$ be an independent decomposition. Then $\mathscr{N}$ is positive dependent if and only if  each subnetwork $\mathscr{N}_i$ is positive dependent.
\end{corollary}

\begin{proof}
($\Rightarrow$) The proof easily follows from the equivalence of the property with the existence of rate constants for each the kinetic system has a positive equilibrium. The set of rate constants for the whole network also provide such for each subnetwork, showing that they each are positive dependent. ($\Leftarrow$) Each subnetwork can be translated to a weakly reversible subnetwork, and their union is weakly reversible. It also forms a decomposition because the translation leaves the set of reaction vectors unchanged. Clearly, this also means that the decomposition is independent. Hence we have weakly reversible translate for the whole network and it follows that it is positive dependent.
\end{proof}

\subsection{TOWARDZ – an algorithm and MATLAB tool for weakly reversible deficiency zero translates of mass action systems}
Hong et al. \cite{HONG2023}  developed a computational framework (called TOWARDZ) for determining whether a network can be translated to a weakly reversible deficiency zero system. The goal of this section is to elaborate on some of aspects of network translation of TOWARDZ as a particularly useful subset of network transformations.

\subsubsection{Some observations about the merge operation}
If two reactions $R_1$ and $R_2$ with the same reaction vector have kinetics $K_1$ and $K_2$, then, to maintain dynamic equivalence, their merge $R'$ has the kinetics $K_1 + K_2$.  If $K_i$ belong to a set $\mathscr{K}$, then their sum may or may not belong to $\mathscr{K}$. Consider the following examples: rate constant-interaction map decomposable kinetics (RIDK), complex factorizable kinetics (CFK) and poly-PL systems (PYK). See Appendix \ref{fundamentals} for precise definitions.
 
\begin{example}
    Let $\mathscr{K} =  RIDK$. Then $K_1 + K_2$ is in RIDK. If $K_1 = k_1I_1$ and $K_2= k_2I_2$ , then  $K_1 + K_2 = k'I'$, with $k' = k_1$ and $I'= I_1 + \frac{k_2}{k_1}I_2$ - the latter is clearly an interaction map.
\end{example}
 
\begin{example}
    Similarly, if $\mathscr{K} =  CFK$. Then $K_1 + K_2$ is in $CFK$. This is because $K_1 + K_2 = \rho'\psi'$, with $\psi'= \psi_1 + \frac{k_2}{k_1} \psi_2.$
\end{example}
 
\begin{example}
    The set of poly-PL kinetics $PYK$ is characterized by the fact that each kinetics can be described as the sum of $PLK$ systems, so clearly the sum of two elements of $PYK$ is again $PYK$.
\end{example}

The 3 examples above fall under the following sufficient condition:
 
\begin{proposition}
    If the set of kinetics $\mathscr{K}$ is an additive semigroup, then the network translate of any kinetic system $(\mathscr{N},K)$ with $K$ from $\mathscr{K}$, also has a kinetics from $\mathscr{K}$.
\end{proposition}

\subsubsection{A partial solution of the ``merge problem''}

However, many kinetic sets do not have the additive semigroup property. In these sets, it depends on the individual pairs of reactions whether the merge operation results in a kinetic from the same set. A sufficient condition is the following:
 
\begin{proposition}
    Let $\mathscr{K}$ be a set of RID kinetics which contains positive multiples of all its elements and $(\mathscr{N},K)$ a kinetic system with $K \in \mathscr{K}$.  If the reactions $q, q'$ have kinetics with identical interaction maps $I = I'$, then their merge's kinetics is also in $K$.
\end{proposition}
 
\begin{example}
    The set of power law kinetics PLK and many of its subsets e.g. PL-RDK have the ``semi-module'' property cited above, and hence ``merging via rate constants'' can be performed. Note however, that there are important subsets, such as factor span surjective kinetics, where reactions with different reactant complexes never have the identical interaction maps. Note in particular that the set of mass action kinetic is a factor span surjective subset of PL-NIK.
\end{example}
 
\begin{example}
    The set of Hill-type kinetics (HTK) is another important set which has the positive semi-module property, although it is also not an additive semigroup.
\end{example}

\section{Equilibria analysis of biochemical systems with positive deficiency network translates}
\label{Sec:3}
In this Section, we show that the examples of olfactory calcium signaling and adipose insulin signaling also provide new insights into equilibria although the zero deficiency property is absent.

\subsection{A mass action system for calcium signaling in olfactory cilia}
\label{4.2.3}
In \cite{SF2012}, Shinar and Feinberg constructed a mass action kinetic realization of a model of calcium signaling in olfactory cilia by Reidl et al. \cite{REIN2006}. \\
\begin{equation}
\nonumber
\begin{split}
R_{1}: & A \rightarrow  B \\
R_{2}: & B \rightarrow A \\ 
R_{3}: & B \rightarrow  D+B  \\
R_{4}: & C+4D \rightarrow E  \\
R_{5}: & E \rightarrow C+4D  \\
\end{split} \quad \quad
\begin{split}
R_{6}: & B+E \rightarrow  F  \\
R_{7}: & F \rightarrow A+E  \\
R_{8}: & A+E \rightarrow F  \\
R_{9}: & D \rightarrow 0.  
\end{split}
\end{equation} 
The network has deficiency $= 10 - 4 - 4 = 2$ and being t-minimal, has the KSSC property. It is also concordant. (See Definitions \ref{ODE:KSSC} and \ref{concordance} in Appendix \ref{fundamentals}). The following transformation produces a weakly reversible realization:

\begin{enumerate}
    \item Shifting $R_9$ by adding $B$ yields $R_{9}: D+B \rightarrow B$.
    \item Dividing $R_1$ yields $R_{1}: A \rightarrow B$ and $R_{10}: A+E \rightarrow B+E$ with rate constants $k_1$ and $k_2$, respectively. Note that $k_1 + k_2 = r_1$.
\end{enumerate}

The weakly reversible kinetic realization is hence the following:
\begin{equation}
\nonumber
\begin{split}
R_{1}: & A \rightarrow  B \\
R_{2}: & B \rightarrow A \\ 
R_{3}: & B \rightarrow  D+B  \\
R_{4}: & C+4D \rightarrow E  \\
R_{5}: & E \rightarrow C+4D  \\
\end{split} \quad \quad
\begin{split}
R_{6}: & B+E \rightarrow  F  \\
R_{7}: & F \rightarrow A+E  \\
R_{8}: & A+E \rightarrow F  \\
R_{9}: & D+B \rightarrow B  \\
R_{10}: & A+E \rightarrow B+E.  \\
\end{split}
\end{equation} \\
The kinetic order matrix is given by:
$$\begin{blockarray}{cccccccc}
  & A & B & C & D & E & F & \text{rate cons.} \\
\begin{block}{c|cccccc|c}
R_1 & 1 & 0 & 0 & 0 & 0 & 0 & k_1 \\
R_2 & 0 & 1 & 0 & 0 & 0 & 0 & r_2 \\
R_3 & 0 & 1 & 0 & 0 & 0 & 0 & r_3 \\
R_4 & 0 & 0 & 1 & 4 & 0 & 0 & r_4 \\
R_5 & 0 & 0 & 0 & 0 & 1 & 0 & r_5 \\
R_6 & 0 & 1 & 0 & 0 & 0 & 0 & r_6 \\
R_7 & 0 & 0 & 0 & 0 & 0 & 1 & r_7 \\
R_8 & 1 & 0 & 0 & 0 & 1 & 0 & r_8 \\
R_9 & 0 & 0 & 1 & 0 & 0 & 0 & r_9 \\
R_{10} & 1 & 0 & 0 & 0 & 0 & 0 & k_2 \\
\end{block}
\end{blockarray}.$$

The transformed system has deficiency $=8-3- 4=1$, and has the KSSC property being weakly reversible and factor span surjective.  The concordance test in CRNToolbox shows that it is discordant. Despite its discordance and positive deficiency, the realization turns out to be useful for the equilibria analysis of the system. We begin our analysis  with the fact that a basis of $\tilde{S}^\perp$ is given by $$\{ [1 \quad 1 \quad 1 \quad \frac{-1}{4} \quad 0 \quad 1 ] ^\top \}.$$

This implies that $\tilde{s} = 5$, hence the transform has kinetic deficiency $= 8 – 3 – 5 = 0$.
It follows from the results in \cite{MURE2014} that the transform is unconditionally complex balanced, i.e., it is complex balanced for any set of rate constants. Being complex balanced and PL-RDK, the transform is also CLP. (See Definition \ref{def:CLP} in Appendix \ref{A2:robustness}). \\

It follows from the species hyperplane criterion in Lao et al. \cite{LLMM2022} that the system has:
\begin{enumerate}
    \item BCR in species E and
    \item No BCR (and consequently, no ACR) in the remaining species.
\end{enumerate}

Two weaknesses of the transform do not allow the conclusion of the PLP property (and consequently ACR in species E):
\begin{enumerate}
\item Its deficiency is non-zero, hence we cannot immediately conclude that it is absolutely complex balanced, and hence by Jose et al. \cite{MTJ2022}, we cannot conclude PLP and apply the species hyperplane criterion to conclude ACR in E.
\item 	Its linkage classes are not independent. Otherwise, we note that, first, that the system is monostationary, i.e., there is at most one positive equilibrium in any stoichiometric class. This claim follows from the fact that the original system has weakly monotonic kinetics (being mass action), and hence injective and furthermore, that it and the transform have identical stoichiometric classes and, second, the system is factor span surjective.
\end{enumerate}
The following result from \cite{HEME2023} could be used:
\begin{proposition}
    Let $(\mathscr{N},K)$ be a cycle terminal PL-FSK system with ILC and non-empty $E_+$. Then $(\mathscr{N},K)$ is multi-PLP implies $(\mathscr{N},K)$ is multi-stationary.
\end{proposition}

\begin{corollary}
    Any complex balanced, monostationary PL-FSK system with ILC is absolutely complex balanced. 
\end{corollary}
 
Nevertheless, it has been useful to establish non-ACR of the remaining 5 species of the system.

\subsection{A mass action system for metabolic insulin signaling in adipocytes}
\label{Sec:MAK_adipocytes}

In \cite{LUML2022}, Luvenia et al. constructed a MAK realization of a model of metabolic insulin signaling in adipocytes by Sedaghat et al. \cite{SEDA2002}.

\begin{equation}
\nonumber
\begin{split}
R_{1} : & X_{2} \rightarrow X_{3} \\
R_{2} : & X_{3} \rightarrow X_{2} \\
R_{3} : & X_{5} \rightarrow X_{4} \\
R_{4} : & X_{4} \rightarrow X_{5} \\
R_{5} : & X_{3} \rightarrow X_{5} \\
R_{6} : & X_{5} \rightarrow X_{2} \\
R_{7} : & X_{2} \rightarrow X_{6} \\
R_{8} : & X_{6} \rightarrow X_{2} \\
R_{9} : & X_{4} \rightarrow X_{7} \\
R_{10} : & X_{7} \rightarrow X_{4} \\
R_{11} : & X_{5} \rightarrow X_{8} \\
R_{12} : & X_{8} \rightarrow X_{5} \\
R_{13} : & 0 \rightarrow X_{6} \\
R_{14} : & X_{6} \rightarrow 0 \\
R_{15} : & X_{7} \rightarrow X_{6} \\
R_{16} : & X_{8} \rightarrow X_{6} \\
R_{17} : & X_{9} + X_{4} \rightarrow X_{10} + X_{4} \\
R_{18} : & X_{9} + X_{5} \rightarrow X_{10} + X_{5} \\
\end{split} \quad \quad
\begin{split}
R_{19} : & X_{10} \rightarrow X_{9} \\
R_{20} : & X_{10} + X_{11} \rightarrow X_{12} \\
R_{21} : & X_{12} \rightarrow X_{10} + X_{11} \\
R_{22} : & X_{14} + X_{12} \rightarrow X_{13} + X_{12} \\
R_{23} : & X_{13} \rightarrow X_{14} \\
R_{24} : & X_{15} \rightarrow X_{13} \\
R_{25} : & X_{13} \rightarrow X_{15} \\
R_{26} : & X_{16} + X_{13} \rightarrow X_{17} + X_{13} \\
R_{27} : & X_{17} \rightarrow X_{16} \\
R_{28} : & X_{18} + X_{13} \rightarrow X_{19} + X_{13} \\
R_{29} : & X_{19} \rightarrow X_{18} \\
R_{30} : & X_{20} \rightarrow X_{21} \\
R_{31} : & X_{21} \rightarrow X_{20} \\
R_{32} : & X_{20} + X_{17} \rightarrow X_{21} + X_{17} \\
R_{33} : & X_{20} + X_{19} \rightarrow X_{21} + X_{19} \\
R_{34} : & 0 \rightarrow X_{20} \\
R_{35} : & X_{20} \rightarrow 0. \\ 
\end{split}
\end{equation} \\

Lubenia et al. \cite{LUML2022} considered an independent deficiency-oriented decomposition $\mathscr{N}=\mathscr{N}_A \cup \mathscr{N}_B$ where
$$\mathscr{N}_A = \{ R_1, \cdots R_{14}, R_{20},R_{21},R_{24},R_{25},R_{34},R_{35} \} \text{ and } \mathscr{N}_B = \mathscr{N} - \mathscr{N}_{A}.$$
Here, $\mathscr{N}_A$ is both weakly reversible and nonconservative, whereas subnetwork $\mathscr{N}_B$ is not weakly reversible and is conservative. They are both t-minimal and concordant, making them monostationary. The following operation produces a weakly reversible realization of $\mathscr{N}_B$:

\begin{enumerate}
    \item Dividing $R_{19}$ yields $R_{19A} : X_{10} + X_4 \rightarrow X_{9} + X_4$ and $R_{19B} : X_{10} + X_5 \rightarrow X_{9} + X_5$.
    \item Shifting $R_{23}$ by adding $X_{12}$ yields $R_{23} : X_{13} + X_{12} \rightarrow X_{14} + X_{12}$.
    \item Shifting $R_{27}$ and $R_{29}$ by adding $X_{13}$ yields $R_{27} : X_{17} + X_{13}  \rightarrow X_{16} + X_{13}$ and $R_{29} : X_{19} + X_{13} \rightarrow X_{18} + X_{13}$.
    \item Dividing $R_{31}$ yields $R_{31A} : X_{21} \rightarrow X_{20}$ and  $R_{31B} : X_{21} + X_{17} \rightarrow X_{20} + X_{17}$ and $R_{31C} : X_{21} + X_{19} \rightarrow X_{20} + X_{19}$.
\end{enumerate}

The weakly reversible kinetic realization is hence the following:

\begin{equation}
\nonumber
\begin{split}
R_{17} : & X_{9} + X_{4} \rightarrow X_{10} + X_{4} \\
R_{18} : & X_{9} + X_{5} \rightarrow X_{10} + X_{5} \\
R_{19A} : & X_{10} + X_4 \rightarrow X_{9} + X_4 \\
R_{19B} : & X_{10} + X_5 \rightarrow X_{9} + X_5 \\
R_{22} : & X_{14} + X_{12} \rightarrow X_{13} + X_{12} \\
R_{23} : & X_{13} + X_{12} \rightarrow X_{14} + X_{12} \\
R_{26} : & X_{16} + X_{13} \rightarrow X_{17} + X_{13} \\
R_{27} : & X_{17} + X_{13}  \rightarrow X_{16} + X_{13}  \\
\end{split} \quad
\begin{split}
R_{28} : & X_{18} + X_{13} \rightarrow X_{19} + X_{13} \\
R_{29} : & X_{19} + X_{13} \rightarrow X_{18} + X_{13} \\
R_{30} : & X_{20} \rightarrow X_{21} \\
R_{31A} : & X_{21} \rightarrow X_{20} \\
R_{31B} : & X_{21} + X_{17} \rightarrow X_{20} + X_{17} \\
R_{31C} : & X_{21} + X_{19} \rightarrow X_{20} + X_{19} \\
R_{32} : & X_{20} + X_{17} \rightarrow X_{21} + X_{17} \\
R_{33} : & X_{20} + X_{19} \rightarrow X_{21} + X_{19} \\
\end{split}
\end{equation} 
The CRNToolbox tests show that both conservativeness and concordance are preserved. Thus, the following Equilibria Theorem of Shinar and Feinberg in \cite{SF2012} applies:
 
\begin{theorem} [Theorem 6.8 of \cite{SF2012}] If $K$ is a continuous kinetics for a conservative reaction network $\mathscr{N}$, then the kinetic system $(\mathscr{N},K)$ has an equilibrium within each stoichiometric compatibility class. If the network is weakly reversible and concordant then within each nontrivial stoichiometric compatibility class there is a positive equilibrium. If, in addition, the kinetics is weakly monotonic, then that positive equilibrium is the only equilibrium in the stoichiometric compatibility class containing it.
\end{theorem}
 
This new result for $\mathscr{N}_B$  allows an improvement of the following Proposition from \cite{LUML2022}:
 
\begin{proposition}
   Let $(\mathscr{N},K)$ be the metabolic insulin signaling network and $\mathscr{N} = \mathscr{N}_A \cup \mathscr{N}_B$ its deficiency-oriented decomposition. Then:

\begin{enumerate}
    \item[i.] The map $\epsilon: E_+ (\mathscr{N},K) \rightarrow \mathbb{R}^\mathscr{S}/S$ given by $\epsilon(x):=x+S$ is injective.
    \item[ii.] The map $\epsilon_A: E_+ (\mathscr{N}_A,K_A) \rightarrow \mathbb{R}^\mathscr{S}/S$ given by $\epsilon_A(x):=x+S$ is surjective.
    \item[iii.] $Az + S \in \Ima (\epsilon)$ iff there is $x \in \epsilon^{-1}_A (z+S)$ such that $p_{\mathscr{CS}}(x) \in p_{\mathscr{CS}}(E_+ (\mathscr{N}_B,K_B))$.  
\end{enumerate}
\end{proposition}

We can add an analogous statement (ii)' as follows: \\

$ii'$. The map $\epsilon_B: E_+ (\mathscr{N}_B,K_B) \rightarrow \mathbb{R}^\mathscr{S}/S$ given by $\epsilon_B(x):=x+S$ is surjective. \\

This additional property implies that the equilibria parametrization for the subnetwork $\mathscr{N}_A$ is no longer necessary and leads to a much shorter proof of the main proposition in the section.

\section{An algorithm for constructing weakly reversible CFK translates for any weakly reversible NFK system NFK systems}
\label{Sec:4}

 We now provide a further method to broaden the use of weakly reversible network translates of positive deficiency. CF-WR translates a weakly reversible NFK (of arbitrary deficiency) to a weakly reversible CFK system (usually with positive deficiency). After a description of the procedure, we show how it significantly improves a result on the running example in Nazareno et al. \cite{NAZA2019}.

\subsection{The CF-WR algorithm for weakly reversible NFK systems}
\label{4.1}
In this Section, we present an algorithm for transforming a weakly reversible NFK system to a weakly reversible CFK system. The procedure is an extension of the CF-transformation introduced by Nazareno et al. in \cite{NAZA2019}. The key idea revolves around the complex which is the reactant of reactions with non-proportional interaction maps, called branching reactions. Our approach involves transforming these reactions by shifting, thereby ensuring the preservation of reaction vectors and the dynamics. In order to preserve weak reversibility, we use the dividing operation to duplicate the cycle that contains the shifted reaction. We recall the following definitions:
 
\begin{definition}
   For a reactant complex $y$ of a network $\mathscr{N}$, $\mathscr{R}(y)$ denotes its set of (branching) reactions, i.e., $\rho^{-1}(y)$ where $\rho : \mathscr{R} \rightarrow \mathscr{C}$ is the \textbf{reactant map}. The $n_r$ reaction sets $\mathscr{R}(y)$ of reactant complexes partition the set of reactions $\mathscr{R}$. Two reactions $r, r' \in \mathscr{R}(y)$ are \textbf{CF-equivalent} for $K$ if their interaction functions coincide, i.e., $I_{K;r} = I_{K;r'}$ or, equivalently, if their kinetic functions $K_r$ and $K_r'$ are proportional (by a positive constant). 
\end{definition}
 
\begin{definition}
    If $N_R(y)$ is the number of CF-subsets of $y$, then $1 \leq N_R(y) \leq |\rho^{-1}(y)|$. The reactant complex is a \textbf{CF-node} if $N_R(y) = 1$, and an \textbf{NF-node} otherwise. It is a \textbf{maximally NF-nod}e if $N_R(y) = |\rho^{-1}(y)| > 1$. The number $N_R$ of CF subsets of a CRN is the sum of $N_R(y)$ over all reactant complexes. If all reactant complexes are CF-nodes for $K$ (i.e., $N_R = n_r$), the kinetics is \textbf{CFK}, otherwise, it is \textbf{NFK}.
\end{definition}
 
\begin{example}
    For a power law kinetics, PL-RDK systems are CFK, since the CF-subsets of a reactant complex are the subsets of branching reactions with identical rows in the kinetic order matrix. Thus, the PL-NDK systems are NFK.
\end{example}
 
\begin{definition}
    For CRN $\mathscr{N}$, a \textbf{cycle} is a sequence of reactions given by:

    $$C_1 \xrightarrow{R_1} C_2  \xrightarrow{R_2} C_3 \xrightarrow{R_3} \cdots \xrightarrow{R_c} C_1$$

    where $C_i$'s are the reactant and product complexes of the given reactions. In the algorithm, we consider simple cycles, i.e., no repeated reactions and no repeated complexes. We denote the subnetwork generated by the simple cycle(s) of reaction(s) $R$ by $\mathscr{N}_{R}^{cycle}$. Note that for all $R \in \mathscr{N}$, $\mathscr{N}_{R}^{cycle}$ exists since $\mathscr{N}$ is weakly reversible.

\end{definition}
 
\begin{procedure}
\label{proc1}
For a NFK kinetic system $(\mathscr{N},K)$, the algorithm proceeds as follows:

\begin{enumerate}
    \item [S1.] Determine the set of NF nodes and its CF-subsets. 

    \item [S2.] Denote $CF_{set}$ as the set of all CF subsets of the NF nodes. 
    
    \begin{enumerate}
    \item [a.] Choose a CF-subset $\in CF_{set}$, denote it $CF_1$. Compute 
    $\mathscr{N}_{CF_1}^{cycle}$.

    \item [b.] Update $CF_{set}$ by removing the CF-subsets covered by the cycle(s) in $\mathscr{N}_{CF_1}^{cycle}$.

    \item [c.] Choose another CF-subset $\in CF_{set}$, denote it $CF_2$. Compute $\mathscr{N}_{CF_2}^{cycle}$ and update $CF_{set}$ by removing the CF-subsets covered by the cycle(s) in $\mathscr{N}_{CF_2}^{cycle}$.

    \item [d.] Continue accordingly until $CF_{set} = \emptyset$. Let $\alpha$ be the number of subnetwork generated by Step S2.
    \end{enumerate}

\item [S3.] Let $\mathscr{R}_{CF_i}$ be the set of reactions of  $\mathscr{N}_{CF_i}^{cycle}$. Define $\mathscr{R}_{set}:= \mathscr{R}-\cap \mathscr{R}_{CF_i}$.
    
    \begin{enumerate}
    \item [a.] Choose a reaction $\in \mathscr{R}_{set}$, denote it $R_1$. Compute 
    $\mathscr{N}_{R_1}^{cycle}$.

    \item [b.] Update $\mathscr{R}_{set}$ by removing the reactions covered by the cycle(s) in $\mathscr{N}_{R_1}^{cycle}$.

    \item [c.] Choose another reaction $\in \mathscr{R}_{set}$, denote it $R_2$. Compute $\mathscr{N}_{R_2}^{cycle}$. and update $\mathscr{R}_{set}$ by removing the reactions covered by the cycle(s) in $\mathscr{N}_{R_2}^{cycle}$.

    \item [d.] Continue accordingly until $\mathscr{R}_{set} = \emptyset$. Let $\beta$ be the number of subnetwork generated by Step S3.
    \end{enumerate}

\item [S4.] Denote the  subnetworks $\mathscr{N}_{CF_1}^{cycle}, \cdots, \mathscr{N}_{CF_\alpha}^{cycle}, \mathscr{N}_{R_1}^{cycle},\cdots, \mathscr{N}_{R_\beta}^{cycle}$ as $\mathscr{N}_1, \cdots, \mathscr{N}_{\alpha+\beta}$. Let $\kappa_{R_{j}}$ be the number of occurrences of $R_j$ in $\{\mathscr{N}_i \}_{1}^{\alpha+\beta}$.

\item [S5.] Determine the reactant set $\rho(\mathscr{R})$. Select a complex $z$ such that for any positive integer $a$, $a \cdot z \notin \rho(\mathscr{R})$  and $a \cdot z \neq y'-y$ where $y \rightarrow y' \in \mathscr{R}$.

    \begin{itemize}
        \item For each subnetwork $\mathscr{N}_i$, we translate each reaction $R_j$ using the operations dividing and shifting by adding $(i-1)z$. That is, the reaction $R_j:y \rightarrow y'$ will yield $R_j^*:y + (i-1)z \rightarrow y' + (i-1)z$ with kinetics 
        
        $$K_{R_j}^*= \frac{k_j}{\kappa_{R_{j}}} I_{K,j}(x),$$
        
        where $\frac{k_j}{\kappa_{R_{j}}}$ is the rate constant and $I_{K,j}(x)$ is the interaction map of reaction $R_j$.
    \end{itemize}
\end{enumerate}

\end{procedure}
    The algorithm will generate a kinetic system $(\mathscr{N}^*,K^*)$ which is dynamically equivalent to $(\mathscr{N},K)$. Because of the numerous choices made within the algorithm, the generated kinetic system is not unique. See Sections \ref{4.2} and \ref{ex2} for illustrations of Procedure \ref{proc1}.

\begin{remark} Procedure \ref{proc1} can generate various transforms based on the selection of the CF-set in Step S2. This selection leads to different network properties. However, since the algorithm only employs shifting and dividing operations, it ensures the preservation of the stoichiometric subspace $S$ and rank $s$. Considering the values of $\alpha$, $\beta$, and $\kappa_{R_{j}}$ from Step S5, we can describe the following network properties:

\begin{itemize}
    \item Linkage classes: $l^* = \alpha + \beta$
    \item Number of reactions and complexes: $r^* = n^* = \sum \kappa_{R_{j}}$
    \item Deficiency: $\delta^* = \sum \kappa_{R_{j}} – (\alpha + \beta) - s$
\end{itemize}
\end{remark}

\subsection{An Improvement of a result of Nazareno et al.}
\label{4.2}
Nazareno et al. use a PL-NDK system for R. Schmitz's model of the earth's pre-industrial carbon cycle as the running example of the paper. Using their generalization of the Johnston-Siegel criterion for linear conjugacy and Mixed Integer Linear Programming (MILP) techniques, they obtained the following result for the Schmitz system: \\
\begin{equation}
  \begin{split}
R_{1}: & M_{1} \rightarrow  M_{2} \\
R_{2}: & M_{1} \rightarrow M_{3} \\ 
R_{3}: & M_1 \rightarrow  M_5  \\
R_{4}: & M_{2} \rightarrow M_1  \\
R_{5}: & M_{2 } \rightarrow M_3  \\
R_{6}: & M_{2} \rightarrow  M_{4}  \\
R_{7}: & M_{3} \rightarrow M_1  \\
  \end{split}
\quad \quad 
  \begin{split}
R_{8}: & M_{3} \rightarrow M_4  \\
R_{9}: & M_{4} \rightarrow M_2  \\
R_{10}: & M_{4} \rightarrow  M_{3}  \\
R_{11}: & M_{5} \rightarrow M_1   \\
R_{12}: & M_{5} \rightarrow M_6  \\
R_{13}: & M_{6} \rightarrow M_1  \\
   \\
  \end{split} 
\quad \quad 
  \begin{split}
F = \begin{blockarray}{ccccccc}
  & X_1 & X_2 & X_3 & X_4 & X_5 & X_6 \\
\begin{block}{c|cccccc|}
R_1	&	1	&	0	&	0	&	0	&	0	&	0	\\
R_2	&	1	&	0	&	0	&	0	&	0	&	0	\\
R_3	&	0.36	&	0	&	0	&	0	&	0	&	0	\\
R_4	&	0	&	9.4	&	0	&	0	&	0	&	0	\\
R_5	&	0	&	1	&	0	&	0	&	0	&	0	\\
R_6	&	0	&	1	&	0	&	0	&	0	&	0	\\
R_7	&	0	&	0	&	10.2	&	0	&	0	&	0	\\
R_8	&	0	&	0	&	1	&	0	&	0	&	0	\\
R_9	&	0	&	0	&	0	&	1	&	0	&	0	\\
R_{10}	&	0	&	0	&	0	&	1	&	0	&	0	\\
R_{11}	&	0	&	0	&	0	&	0	&	1	&	0	\\
R_{12}	&	0	&	0	&	0	&	0	&	1	&	0	\\
R_{13}	&	0	&	0	&	0	&	0	&	0	&	1	\\
\end{block}
\end{blockarray}
  \end{split} 
  \nonumber
\end{equation}

There exists a weakly reversible deficiency one PL-RDK system $(\mathscr{N}^*, K^*)$ and rate constants such that $(\mathscr{N}, K)$  is linearly conjugate to $(\mathscr{N}^*, K^*)$  for the rate constant values. The conjugacy vector is $(2.28,1.14,1.14,1.14,4.56,4.56)$, so the conjugacy is not a dynamical equivalence. \\

We use the CF-WR algorithm to find the dynamically equivalent weakly reversible PL-RDK system.
\begin{enumerate}
    \item[S1.] The NDK nodes are $M_1$,  $M_2$ and $M_{3}$. See Table \ref{Schmitz1}. \\

        \begin{table}[h!]
        \centering
        \caption{CF-subsets of Schmitz}\label{Schmitz1}%
        \begin{tabular}{llcc}
        \hline
        NF nodes & Reaction set & CF-subsets \\
        \hline
        $M_1$ & $R_{1},R_{2},R_{3}$ & $ \{R_{1}, R_{2}  \} ,\{ R_{3}  \}$ \\
        $M_2$  & $R_{4},R_{5},R_{6}$ & $\{R_{4}\},\{R_{5},R_{6}  \} $ \\
        $M_{3}$ & $R_{7},R_{8}$ & $\{R_{7}  \},\{R_{8}  \}$ \\

        \hline
        \hline
        \end{tabular}
        \end{table}

\item[S2.] Hence, we have
$$CF_{set} = \{ 
\{R_{1}, R_{2}  \} ,\{ R_{3}  \},
\{R_{4}\},\{R_{5},R_{6}  \}, 
\{R_{7}  \},\{R_{8}  \}
\}.$$

    \begin{itemize}
        \item For the 1st iteration of S2, choose CF subset $\{R_{3}  \}$. Here, we compute

        $$\mathscr{N}_{CF_1}^{\text{cycle}} = \{R_{i} \text{ where } i \in \{3,11,12,13\} \}$$ 

        $$CF_{set} = \{ 
\{R_{1}, R_{2}  \},
\{R_{4}\},\{R_{5},R_{6}  \}, 
\{R_{7}  \},\{R_{8}  \}
\}.$$

        \item For the 2nd iteration of S2, choose CF subset $\{R_{1},\{R_{2}  \}$. Here, we compute

        $$\mathscr{N}_{CF_2}^{\text{cycle}} = \{R_{i} \text{ where } i \in \{1,2,4,7 \} \}$$

        $$CF_{set} = \{ 
\{R_{5},R_{6}  \}, \{R_{8}  \}
\}.$$
        \item For the 3rd iteration of S2, choose CF subset $\{R_{5},R_{6}  \}$. Here, we compute

        $$\mathscr{N}_{CF_3}^{\text{cycle}} = \{R_{i} \text{ where } i \in \{5,6,8,9,10 \} \}$$ 

        \item We update the $CF_{set} = \emptyset.$

\end{itemize}
\item[S3.] We compute for $\mathscr{R}_{set}$, the set of unassigned reactions.

$$\mathscr{R}_{set}=\emptyset$$

\item[S4.] For brevity, we denote the  subnetworks $\mathscr{N}_{CF_1}^{cycle}$, $\mathscr{N}_{CF_2}^{cycle}$, $\mathscr{N}_{CF_3}^{cycle}$ as $\mathscr{N}_1, \mathscr{N}_2, \mathscr{N}_{3}$, respectively.\\ 
\item [S5.] The reactant set is
$\rho(\mathscr{R})= \{ M_i \text{ for } i \in \{1,2, \cdots 6 \} \}$. We select the complex $M_1$ for the following steps:
    \begin{itemize}
        \item For subnetwork $\mathscr{N}_1$, we translate each reaction $R_j$ using the operations dividing and shifting by adding $(0)M_1$. Thus, $\mathscr{N}^* = \mathscr{N}_1$. We have the following translated reactions
        $$\begin{array}{l}
            R_{3}^*: M_1 \rightarrow  M_5  \\
            R_{11}^*: M_{5} \rightarrow M_1   \\
            R_{12}^*: M_{5} \rightarrow M_6  \\
            R_{13}^*: M_{6} \rightarrow M_1.  \\
            \end{array}$$ \\
           \item For subnetwork $\mathscr{N}_2$, we translate each reaction $R_j$ using the operations dividing and shifting by adding $(1)M_1$. Thus, $\mathscr{N}^*_2$ have the following translated reactions
        $$\begin{array}{l}
R_{1}^*: M_{1} + (1)M_1 \rightarrow  M_{2} + (1)M_1\\
R_{2}^*: M_{1} + (1)M_1 \rightarrow M_{3} + (1)M_1 \\ 
R_{4}^*: M_{2} + (1)M_1 \rightarrow M_1 + (1)M_1 \\
R_{7}^*: M_{3} + (1)M_1 \rightarrow M_1 + (1)M_1 \\
            \end{array}.$$ \\
           \item For subnetwork $\mathscr{N}_3$, we translate each reaction $R_j$ using the operations dividing and shifting by adding $(2)M_1$. Thus, $\mathscr{N}^*_3$ have the following translated reactions
        $$\begin{array}{l}
R_{5}^*: M_{2 } + (2)M_1 \rightarrow M_3 + (2)M_1  \\
R_{6}^*: M_{2} + (2)M_1 \rightarrow  M_{4} + (2)M_1  \\
R_{8}^*: M_{3} + (2)M_1 \rightarrow M_4 + (2)M_1  \\
R_{9}^*: M_{4} + (2)M_1 \rightarrow M_2 + (2)M_1  \\
R_{10}^*: M_{4} + (2)M_1 \rightarrow  M_{3} + (2)M_1  \\
            \end{array}.$$
\item Lastly, we have $\mathscr{N}^* = \cup \mathscr{N}^*_i$.
    \end{itemize}
\end{enumerate}         
Thus, we obtain the following improvement: for all rate constants, the Schmitz system is dynamically equivalent to the weakly reversible deficiency one PL-RDK system:
\begin{equation}
  \begin{split}
R_{1}^*: & 2M_{1} \rightarrow  M_{2} + M_1 \\
R_{2}^*: & 2M_{1} \rightarrow M_{3} + M_1 \\ 
R_{3}^*: & M_1 \rightarrow  M_5  \\
R_{4}^*: & M_{2} + M_1 \rightarrow 2M_1  \\
R_{5}^*: & M_{2 } + 2M_1 \rightarrow M_3 + 2M_1  \\
R_{6}^*: & M_{2} + 2M_1 \rightarrow  M_{4} + 2M_1  \\
R_{7}^*: & M_{3} + M_1 \rightarrow 2M_1  \\
  \end{split}
\quad \quad 
  \begin{split}
R_{8}^*: & M_{3} + 2M_1 \rightarrow M_4 + 2M_1 \\
R_{9}^*: & M_{4} + 2M_1 \rightarrow M_2 + 2M_1  \\
R_{10}^*: & M_{4} + 2M_1 \rightarrow  M_{3} + 2M_1  \\
R_{11}^*: & M_{5} \rightarrow M_1   \\
R_{12}^*: & M_{5} \rightarrow M_6  \\
R_{13}^*: & M_{6} \rightarrow M_1  \\
   \\
  \end{split} 
  \nonumber
\end{equation}

\section{CF-WR Applications to Interaction Span Surjective NFK (ISK) systems}
\label{Sec:5}

We derive two new results on weakly reversible NFK systems using the CF-WR translation to further demonstrate its utility. The first Proposition enhances a Subspace Coincidence Theorem for NFK systems derived in Nazareno et al. \cite{NAZA2019}. The second extends a log parametrization theorem of Hernandez and Mendoza in \cite{HEME2023} for weakly reversible PL-RDK to corresponding PL-NDK systems.

\subsection{An Enhancement of the Subspace Coincidence Theorem for weakly reversible ISS NFK systems}

The containment of a system's kinetic subspace $K$ in the stoichiometric subspace $S$ of its underlying network expresses an important connection between system behavior and network structure. $K \neq S$ or their non-coincidence, for example, implies that all of the system's equilibria are degenerate. M. Feinberg and F. Horn were the first to study the coincidence of the kinetic and the stoichiometric subspaces in 1977 \cite{FEHO1977}. Their main result, which we call the Feinberg-Horn KSSC (Kinetic and Stoichiometric Subspace Coincidence) Theorem, identifies network properties sufficient for coincidence in MAK systems.

\begin{theorem}
    If $K$ is the kinetic subspace of a MAK system, and:
    \begin{enumerate}
        \item [(i)] if $t-l = 0$, then $K=S$.
        \item [(ii)] if $t-l > \delta$, then $K \neq S$.
        \item [(iii)] if $0<t-l \leq \delta$, $K=S$ or $K \neq S$ is rate constant dependent.  
    \end{enumerate}
\end{theorem}

Forty years later, Arceo et al. \cite{AJLM2017} extended the KSSC Theorem on the set of complex factorizable kinetics (CFK) using the concept of span surjectivity as follows:

\begin{theorem}
    Consider a complex factorizable system on a network $\mathscr{N}$.
    \begin{enumerate}
        \item [(i)] If $t-l > \delta$, then $K \neq S$.
        \item [(ii)] If $0< t-l \leq \delta$, and a positive steady state exists, then $K \neq S$. In fact, $\dim S - \dim K \geq t-l-\delta+1$.
    \end{enumerate}

    If the system is also factor span surjective and: 
    \begin{enumerate}
        \item [(iii)] $t-l = 0$, then $K = S$.
        \item [(iv)] $0< t-l \leq \delta$, then it i rate constant dependent whether $K=S$ or not.
    \end{enumerate}
    
\end{theorem}

The key to the generalization was the identification of the CFK subset of factor span surjective kinetics (FSK):

\begin{definition}
If $V,W$ are finite dimensional vector spaces, a map $f: V \rightarrow W$ is \textbf{span surjective} iff $\langle \Ima f \rangle = W$. A complex factorizable kinetics K is \textbf{factor span surjective} if its factor map is span surjective
\end{definition}

Arceo et al. \cite{AJLM2017} characterized a factor span surjective PL-RDK system, as follows:

\begin{proposition}
    A PL-RDK system is factor span surjective iff all rows with different reactant complexes in the kinetic order matrix $F$ are pairwise different.
\end{proposition}

Using the concept of CF-subsets, Nazareno et al. in 2019 introduced the corresponding property for NFK systems: \\

\begin{definition}
    A RID kinetics is \textbf{interaction span surjective} if and only if the set $I_K(\mathscr{R}_i)$ of its CF-subset interaction maps is linearly independent. Here, $I_K(\mathscr{R}_i)$ is the common interaction map of the kinetics of the reactions of the CF-subset $\mathscr{R}_i$.
\end{definition}


\begin{proposition}
    \label{prop13}
    If $(\mathscr{N}, K)$ is interaction span surjective, then its CF-transform $(\mathscr{N}^*, K^*)$ is also factor span surjective.
\end{proposition} 

We denote the subset of interaction span surjective NFK kinetics on a network $\mathscr{N}$ with $ISK(\mathscr{N})$. An important result in Nazareno et al is the following Subspace Coincidence Theorem for NDK systems

\begin{theorem}
    Let $(\mathscr{N}, K)$ be a PL-NDK system.
    \begin{enumerate}
        \item[i)] If $N_R < s$, then $K \neq S$.
    \end{enumerate}

If the system is also interaction span surjective (ISS), then either
    
    \begin{enumerate}
        \item[ii)]  $\mathscr{N}$ is t-minimal and $r-r_{mcf} = N_R - n_r$, implies $K=S$; or
        \item[iii)] $\mathscr{N}$ is TBD and point terminal, implies that $K=S$ is rate-constant dependent.
    \end{enumerate}
\end{theorem}

The following Proposition enhances the previous Theorem and can be added as an additional statement ii') :

\begin{proposition}
    Let $(\mathscr{N}, K)$ be an NFK system. If the system is also interaction span surjective, then $\mathscr{N}$ is weakly reversible implies $K = S$.
\end{proposition}

\begin{proof}
As described in Section \ref{4.1}, the first 4 steps of CF-WR identify the reactions to which shifting operations need to be applied. Just as in CF-RM, a shift with a reactant multiple is carried out for each of them. Hence, it follows from Proposition \ref{prop13} that $(\mathscr{N}^\sharp, K^\sharp)$ is factor span surjective. Since the network is weakly reversible, it is also t-minimal and $K_\omega= S_\omega$. Furthermore, we have $K = K_\omega$ since the kinetics coincide and, because shifts leave the reaction vectors invariant, $S = S_\omega$. Therefore, $K = S$.
\end{proof}   

\begin{remark} 
Combining the main statements of the Subspace Coincidence Theorems, we have the following nice comparison:
\begin{enumerate}
    \item For a CFK, if the network is t-minimal and the kinetics factor span surjective, then $K = S$.
    \item For an NFK system, if the network is weakly reversible and the kinetics interaction span surjective, then $K = S$.
\end{enumerate}
In other words, the higher kinetics complexity of NFK needs to be compensated by a stronger network property  to achieve the coincidence.
\end{remark}

\begin{example}
    Since the Schmitz system is weakly reversible and $ISK$, $K = S$. This fact has not been documented, although the network-specific condition also holds with $r = 13$, $r_m = 10$, $N_R = 9$ and $n_r = 6$.
\end{example}

\subsection{A Log Parametrization (LP) Theorem for weakly reversible PL-ISK systems}

We first recall some concepts:

\begin{definition}
 A kinetic system $(\mathscr{N}, K)$ is a \textbf{poly-PLP system} if its set of positive equilibria $E_+(\mathscr{N}, K)$ is the disjoint union LP sets with flux subspace $P_E$ and reference points $\{ x_i^* | i = 1, \cdots, \mu \}$ where $\mu := | E_+ \cap Q |$ (with $Q$ as a flux class or $\infty$ is the coset intersection count. Analogously, a system is \textbf{poly-CLP} if its set of complex balanced equilibria $Z_+(\mathscr{N}, K)$ is the disjoint union of LP sets. A poly-PLP (poly-CLP) system is called multi-PLP (multi-CLP) if $\mu \geq 2$, otherwise PLP (CLP). We also denote a poly-PLP (poly-CLP) system with $\mu$-PLP ($\mu$-CLP) if we want to emphasize the value of $\mu$.  
\end{definition}

In \cite{HEME2023}, Hernandez and Mendoza studied cycle terminal PL-FSK systems and derived the following Theorem, extending a result of Boros for mass action systems:

\begin{theorem}
Let $(\mathscr{N}, K)$ be a cycle terminal PL-FSK system that satisfies $\delta = \delta_1 +
\cdots + \delta_l$. Suppose $E_+ \neq \emptyset$, $P$ a positive stoichiometric class such that $E_+ \cap P \neq \emptyset$ and for $x^* \in \mathbb{R}^m_{>0}$ and let

$$Q(x^*) = \{ x \in \mathbb{R}^m_{>0} | \log (x) - \log(x^*) \in \tilde{S}^{\perp} \}$$

then $(\mathscr{N}, K)$  is a poly-PLP system with flux subspace $P_E = \tilde{S}$ and reference points of the form $x^{*,j}$, i.e., 
\[
    E_+ = \bigcup_{x^* \in E_+ \cap P} Q(x^*).
\]

\end{theorem}

A Corollary says that if $(\mathscr{N}, K)$  is monostationary, then it is in fact a (mono-) PLP system.  Combining this result with CF-WR, we obtain the following new result:

\begin{theorem}
Let $(\mathscr{N}, K)$ be a weakly reversible ISS PL-NDK system and $(\mathscr{N}_\omega, K_\omega)$ its CF-WR translate. Then

\begin{enumerate}
\item[i.] If $\mathscr{N}_\omega$ has independent linkage classes (ILC), then $(\mathscr{N}, K)$ is a poly-PLP system and

\item[ii.] If, in addition, $(\mathscr{N}, K)$ is mono-stationary, then it is a (mono-) PLP system.
\end{enumerate}
\end{theorem}

\begin{proof}
Since the kinetics of $\mathscr{N}$ is interaction span surjective, it is factor span surjective on $\mathscr{N}_\omega$.  Since$\mathscr{N}_\omega$ is weakly reversible, it is cycle terminal. By assumption, it has ILC. Hence, $(\mathscr{N}_\omega, K_\omega)$ is a poly-PLP system. Due to their dynamical equivalence, their equilibria sets coincide, showing that $(\mathscr{N}, K)$ too is also poly-PLP. By the Corollary, if the system is mono-stationary, then it is (mono-) PLP.
\end{proof}

\begin{remark}
In the case of ii), the ACR species of the system are easy to determine since the Species Hyperplane Criterion for PLP systems holds \cite{LLMM2022}. And thus, $K = S$.
\end{remark}

\begin{example}
 Consider the following subnetwork of the Schmitz network: \\
\begin{center}
\begin{tikzpicture}[baseline=(current  bounding  box.center)]
\tikzset{vertex/.style = {draw=none,fill=none}}
\tikzset{edge/.style = {bend right,->,> = latex', line width=0.20mm}} 
\node[vertex] (1) at  (0,1.5) {$M_5$};
\node[vertex] (2) at  (4,1.5) {$M_2$};
\node[vertex] (3) at  (2,0) {$M_1$};
\node[vertex] (4) at  (6,0) {$M_4$};
\node[vertex] (5) at  (0,-1.5) {$M_6$};
\node[vertex] (6) at  (4,-1.5) {$M_3$};
\draw [edge]  (1) to[""] (3);
\draw [edge]  (3) to[""] (1);
\draw [edge]  (1) to[""] (5);
\draw [edge]  (5) to[""] (3);

\draw [edge]  (2) to[""] (3);
\draw [edge]  (4) to[""] (2);
\draw [edge]  (6) to[""] (4);
\draw [edge]  (3) to[""] (6);
\end{tikzpicture}
\end{center}

It has one NDK node $M_1$, and its CF-WR translate $\mathscr{N}^\sharp$ is obtained by shifting each reaction in the right cycle by $M_1$. $\mathscr{N}^\sharp$ has two linkage classes which are clearly independent. Furthermore, since the network is concordant and the kinetics weakly monotonic, it is mono-stationary. So that the PL-NDK subsystem fulfills all the assumptions of the Theorem and hence is a PLP system. 
\end{example}

\begin{remark}
    The log parametrization of the Schmitz subnetwork was first shown in 2019 by Fortun et al. \cite{FLRM2019}, where it was viewed as a Deficiency Zero Theorem for a class of PL-NDK systems.
\end{remark}

\section{From network translations to network transformations}
\label{Sec:6}

\subsection{The extended set of network operations}

In Section \ref{Sec:2}, we introduce the three (reaction) network operations that translate a kinetic system $(\mathscr{N}, K)$ to $(\mathscr{N}', K)'$. Some of the consequences of this translation are: ODE systems are identical, the sets $E_+(\mathscr{N}, K)$ and $E_+(\mathscr{N}', K')$ are the same, and the kinetic subspaces $\underline{K}$ and $\underline{K'}$ coincide. Here, we extend the set of network operations that generate dynamically equivalent systems.
 
\begin{definition}[Scaling and Splitting operations]
Let $q: y \xrightarrow{r} y'$ be a reaction with kinetics $K_{q}=rI_K$ in a kinetic system $(\mathscr{N}, K)$.
\begin{enumerate}
    \item \textbf{Scaling} the reaction results to  $q':my \xrightarrow{r/m} my'$ with kinetics $K_{q'}=\frac{r}{m}I_K$.
    
    \item \textbf{Splitting (via rate constant difference)} the reaction results to the pair reactions $q': y + a \xrightarrow{r_1} y' + a$ and $q'': y' + b \xrightarrow{r_2} y + b$ with kinetics $K_{q'}=r_1 I_{K}$ and $K_{q''}=r_2 I_{K}$, respectively. Here, $r =r_1 - r_2$, $a = a_1 X_1 + ... + a_m X_m$, $b = b_1 X_1 + ... + b_m X_m$ and $a_i$'s, $b_i$'s $\in \mathbb{Z}$.

    \item \textbf{Splitting (via reaction vector)} the reaction results to the pair reactions $q': x \xrightarrow{r} x' $ and $q'': z \xrightarrow{r} z'$ with the same kinetics (that is, $K_{q}=K_{q'}=K_{q''}$) and $y'-y = (x'-x) + (z'-z)$.
\end{enumerate}
 
\end{definition}
 

\begin{definition}[Merging via reaction vector and adding dummy reactions]
Let $q_1: y \xrightarrow{r_1} y'$ and $q_2: z \xrightarrow{r_2} z'$ be reactions with the same kinetics $K_{q}$.
\begin{enumerate}
    \item \textbf{Merging (via reaction vector)} results to the reaction $q': x \xrightarrow{r} x'$ with kinetics $K_{q}$ where $x'-x = (y'-y) + (z'-z)$.

    \item \textbf{Adding dummy reactions} results to additional pair reaction $p': v \xrightarrow{r} v'$ and $p'': w \xrightarrow{r} w'$ with the same arbitrary kinetics $K_{p}$ where $(v'-v) + (w'-w) = 0$.
\end{enumerate}
\end{definition}

For examples, see Table \ref{tabexx}. It is important to note that a finite number of summands (instead of just two) are allowed in splitting/merging by reaction vector. 

\begin{table}[h]
\caption{Examples for the additional (reaction) network operations}
\label{tabexx}
\begin{tabular*}{\textwidth}{@{\extracolsep\fill} p{0.12\linewidth} llll}
\toprule%
& \multicolumn{2}{@{}c@{}}{Original Kinetic System} & \multicolumn{2}{@{}c@{}}{WR Dynamic Equivalent System} \\\cmidrule{2-3}\cmidrule{4-5}%
Operation & Network & Kinetics & Network & Kinetics \\

\midrule
Scaling &
$\begin{array}{l}
R_1: X \rightarrow Y\\
R_2: 2Y \rightarrow 2X
\end{array}$
 & $\begin{array}{l}
K_1 =r_1 I_{K,1}  \\
K_2 =r_2 I_{K,2}
\end{array}$ & $\begin{array}{l}
R_1^*: 2X \rightarrow 2Y \\
R_2: 2Y \rightarrow 2X
\end{array}$ & $\begin{array}{l}
K_1 =\frac{r_1}{2} I_{K,1}  \\
K_2 =r_2 I_{K,2}
\end{array}$ \\
\midrule
Splitting via rate const. diff. &
$\begin{array}{l}
R_1: 2X \rightarrow X+Y\\
R_2: X+Y \rightarrow 2Y\\
R_3: 2Y \rightarrow X+Y
\end{array}$
 & $\begin{array}{l}
K_1 =r_1 I_{K,1}  \\
K_2 =r_2 I_{K,2} \\
K_3 =r_3 I_{K,3}
\end{array}$ & $\begin{array}{l}
R_1: 2X \rightarrow X+Y\\
R_{2,1}^*: X+Y \rightarrow 2Y\\
R_{2,2}^*: X+Y \rightarrow 2X\\
R_3: 2Y \rightarrow X+Y
\end{array}$
 & $\begin{array}{l}
K_1 =r_1 I_{K,1}  \\
K_{2,1} =r_{21} I_{K,2} \\
K_{2,2} =r_{22} I_{K,2} \\
K_3 =r_3 I_{K,3}
\end{array}$ \\

\midrule
Splitting via react. vect. &
$\begin{array}{l}
R_1: X+Y \rightarrow 2Y\\
R_2: Y \rightarrow X
\end{array}$
 & $\begin{array}{l}
K_1 =r_1 I_{K,1}  \\
K_2 =r_2 I_{K,2}
\end{array}$ & $\begin{array}{l}
R_1: X+Y \rightarrow 2Y\\
R_2^*: 2Y \rightarrow 2X + Y \\
R_3^*: 2X + Y \rightarrow X + Y
\end{array}$& $\begin{array}{l}
K_1 =r_1 I_{K,1}  \\
K_2 =r_2 I_{K,2} \\
K_3 =r_2 I_{K,2}
\end{array}$ \\
\midrule
Merging via react. vect. &
$\begin{array}{l}
R_1: X+Y \rightarrow 2Y \\
R_2: 2X \rightarrow 2Y \\
R_3: 4Y \rightarrow 3X+Y \\
\end{array}$
 & $\begin{array}{l}
K_1 =r_1 I_{K,1}  \\
K_2 =r_1 I_{K,1} \\
K_3 =r_3 I_{K,3}  \\ 
\end{array}$ & $\begin{array}{l}
R_1^*: 3X+Y \rightarrow 4Y \\
R_3: 4Y \rightarrow 3X+Y \\
\end{array}$
 & $\begin{array}{l}
K_1 =r_1 I_{K,1} \\
K_3 =r_3 I_{K,3}  \\ 
\end{array}$ \\

\midrule

Adding dummy react. &
$\begin{array}{l}
R_1: X+Y \rightarrow Z \\
R_2: Z \rightarrow 2Z \\
R_3: X+2Z \rightarrow Y \\
R_4: Y \rightarrow 2X+Y \\
\end{array}$
 & $\begin{array}{l}
K_1 =r_1 I_{K,1}  \\
K_2 =r_2 I_{K,3} \\
K_3 =r_3 I_{K,3}  \\ 
K_4 =r_4 I_{K,4}  \\ 
\end{array}$ & $\begin{array}{l}
R_1: X+Y \rightarrow Z \\
R_2: Z \rightarrow 2Z \\
R_3: X+2Z \rightarrow Y \\
R_4: Y \rightarrow 2X+Y \\
R_5^*: 2Z \rightarrow X+Y \\
R_6^*: 2X+Y \rightarrow X+2Z \\
\end{array}$
 & $\begin{array}{l}
K_1 =r_1 I_{K,1}  \\
K_2 =r_2 I_{K,3} \\
K_3 =r_3 I_{K,3}  \\ 
K_4 =r_4 I_{K,4}  \\ 
K_5 =r I_{K,a}  \\ 
K_6 =r I_{K,a}  \\ 
\end{array}$\\

\bottomrule
\end{tabular*}

\end{table}

We identify two classes of operations depending on whether they leave the set of reaction vectors of the network invariant or not. See Table \ref{tabclass}. 

\begin{table}[h!]
\centering
\caption{Classifications of the network operations}\label{tabclass}%
\begin{tabular}{lcc}
\hline
Network Operations & Reaction Vector Invariant & Stoichiometric Subspace Invariant \\
\hline
\hline
Shifting & YES & YES \\
\hline
Scaling & NO & YES \\
\hline
Dividing & YES & YES \\
\hline
Merging & YES & YES \\
\hline
Splitting via rate const. diff. & NO & YES \\
\hline
Splitting via reac. vect. & NO & NO \\
\hline
Merging via reac. vect. & NO & NO \\
\hline
Adding dummy reactions & NO & NO \\
\hline
\hline
\end{tabular}
\end{table}

There are 3 operations which preserve the set of reaction vectors of the network: shifting, dividing and merging. These form one class which will be called network translation operations. Another classification is given by the criterion whether they leave the stoichiometric subspace of the network or not. In this case, 5 operations are $S$-invariant: the  3 network translation operations, scaling and splitting by rate constant difference. These will be called network transformation operations. Additionally, the operations, splitting via reaction vectors and adding dummy reactions enlarge the original subspace, that is, $S \subseteq S'$ while the operation merging via reaction vectors makes $S' \subseteq S$. From these observations, we have the following chain of subsets of (reaction) network operation for a kinetic system $(\mathscr{N}, K)$ with the same set of species: network operations, network transformation operation and network translation operation.

\subsection{Relationships between network operations}

Let $(\mathscr{N}, K)$ and $(\mathscr{N}', K')$ be dynamically equivalent kinetic systems, i.e., $\mathscr{S} = \mathscr{S}'$ and $f = NK = N'K'= f'$. If $\mathscr{R}$ and $\mathscr{R}'$ are their reaction sets, we have the following definition:

\begin{definition}
    Two kinetics $K_q, K_q'$ for reactions $q, q' \in \mathscr{R} \cup \mathscr{R}'$ are \textbf{equivalent} if their ratio for all $x$ in their domain, is a positive real number. 
    The reactions $q, q'$ are \textbf{kinetically equivalent} if their kinetics are equivalent.
\end{definition}

This concept extends previously introduced equivalences for branching reactions (leading to the concept of CF-subsets) and reactions in a kinetic system in general (leading to the concept of interaction span surjectivity).

Any splitting operation assigns to a reaction $q: y \rightarrow y'$  with kinetics $K_q$ in $\mathscr{N}$ two reactions $q_1$ and $q_2$ in $\mathscr{N}^*$ with kinetics $K_1$ and $K_2$ respectively. The following table shows the relationships between these objects:

\begin{table}[h!]
\caption{Relationship of the kinetics of split reactions}\label{tab1}%
\begin{tabularx}{\linewidth}{X|X|X|X|X}
\hline
Operation & Kinetics of $q$ & Kinetics of $q_1$ & Kinetics of $q_2$ & Requirements \\
\hline
\hline
Dividing (rate constant sum splitting) & 
$K = rI_K$ & $K_1 = \lambda K = (\lambda r)I_K  = r (\lambda I_K)$ & $K_2 = (1-\lambda)K = ((1-\lambda)r)I_K  = r ((1-\lambda) I_K)$ & 
$\begin{array}{l}
  0 < \lambda < 1, \\
   0 < r
\end{array}$
 \\
\hline
Rate constant difference splitting & $K = rI_K$ & $K_1 = r_1 I_K$ & $K_2 = r_2 I_K$ & 
$\begin{array}{l}
  r = r_1 – r_2,  \\
  r_1 > r_2 > 0 
\end{array}$ 
\\
\hline
Reaction vector splitting & $K = rI_K$ & $K_1 = rI_K$ & $K_2 = K = rI_K$ & $r > 0$ \\
\hline
Adding dummy pair & $K = rI_K$ &$ K_1 = rI_K$ & $K_2 = K = rI_K$ & = Extension of reaction vector splitting to zero vector or trivial reaction $y \rightarrow y$  \\
\hline
\hline
\end{tabularx}
\end{table}

\begin{remark}
    Note that ``dividing'' can be viewed as a special case of ``reaction vector splitting'' by taking $x_i – x = \lambda (y'- y)$ and $z'- z = (1- \lambda )(y'- y)$ since $\lambda K(y'- y) = K(x'- x)$ and $((1- \lambda )K)(y'- y) = K (z'- z)$.
\end{remark}
 
We can relate the two merging operations ``reaction vector merge'' and ``merge'' (from Hong et al. \cite{HONG2023}) as follows: recall that ``merge'' assigns to two reactions with  kinetics $K_1$ and $K_2$ a new reaction with the same reaction vector and kinetics $K_1 + K_2$. We can interpret this in the context of ``reaction vector merging'' as follows: let $y'- y$ be the reaction vector of the ``merge'' and set $x'- x = z'- z = (y'- y)/2$.  For the  kinetics $K =(K_1 + K_2)/2$,  we have $K(y'- y) = K_1(x'- x) + K_2(z'- x)$. This generalizes ``reaction vector mering'' to arbitrary kinetics in the above special case of reaction vectors. The ``merge'' kinetics is twice ``reaction vector merging'' kinetics, making them equivalent.  \\

This can be generalized to the case when the reaction vectors are positively proportional and leads to the following generalization of the merge and reaction vector merge operations:
 
\begin{definition}
    Let the $x \rightarrow x'$ and $z \rightarrow z'$ be reactions with kinetics $K_1$ and $K_2$ respectively. Let their reaction vectors be $x'- x = \lambda(z'- z)$, $\lambda > 0$. Their ``reaction vector merge'' is the reaction $y \rightarrow y'$ with $y'- y = (x' - x) + (z'- z)$ and kinetics $K = \frac{\lambda K_1 + K_2}{\lambda+1}$. The kinetics $(\lambda K_1 + K_2)$ is called the \textbf{$\lambda$-merge} of the reactions and generalizes the merge reaction ($\lambda$ = 1).
\end{definition}

The proportionality relation implies that $x'- x =\frac{\lambda}{\lambda+ 1}(y'-y)$ and $z'- z = \frac{1}{\lambda+ 1}(y'-y)$. Multiplying each term with their respective kinetics delivers the claim.

\subsection{Concept and classes of network transformations}

A network transformation of kinetic systems is a dynamic equivalence “aligned" with the stoichiometric structure of the underlying networks. We study 3 classes: S-contracting, S-expanding and their intersection, S-invariant transformations.

\begin{remark}
The term “network transformation” was first used in Nazareno et al. only for S-invariant transformations.
\end{remark}

We now broaden the concept as follows:
 
\begin{definition}
    A dynamic equivalence between $(\mathscr{N},K)$ and $(\mathscr{N}',K')$ is called 
    \begin{enumerate}
        \item[i.] \textbf{S-extending transformation} if $S \subset S'$, 
        \item[ii.] \textbf{S-including transformation} if $S' \subset S$, 
        \item[iii.] \textbf{S-invariant transformation} if $S$ = $S'$ and
        \item[iv.] \textbf{Network transformation} (or simply transformation) if one of i), ii) or iii) holds.
    \end{enumerate}
\end{definition}

 $(\mathscr{N}',K')$ is correspondingly called an S-extension, S-inclusion, S-invariance of $(\mathscr{N},K)$ and network transform (or transform) of $(\mathscr{N},K)$. 
We have the following chain of containment for any kinetic system $(\mathscr{N},K)$: 
\begin{itemize}
    \item the set of all kinetic systems with the same set of species (denoted by $SID(\mathscr{N},K)$), 
    \item its subset of realizations of $(\mathscr{N},K)$, i.e., all systems dynamically equivalent to $(\mathscr{N},K)$ (denoted by $REL(\mathscr{N},K)$) and 
    \item the subset of network transformations of $(\mathscr{N},K)$ (denoted by $TRF(\mathscr{N},K)$).
\end{itemize}

$TRF_>(\mathscr{N},K)$, $TRF_=(\mathscr{N},K)$ and $TRF_<(\mathscr{N},K)$ denote the subsets of S-extensions, S-invariances and S-inclusions respectively.

\begin{remark}
    The network translations, first studied by M. Johnston in 2014 \cite{JOHN2014} and most recently by Hong et al. \cite{HONG2023} are S-invariant transformations with the special property that they leave the set of reaction vectors unchanged.
\end{remark}

We introduce several sets of kinetic systems generated by network operations on $(\mathscr{N},K)$.

\begin{definition}
    Let $(\mathscr{N},K)$ be a kinetic system. 
    \begin{enumerate}
        \item [i.] $MNO_>(\mathscr{N},K) :=$ the set of kinetic systems generated by the S-invariant and S-extending operations.
        \item [ii.] $MNO_=(\mathscr{N},K) :=$ the set of kinetic systems generated by the S-invariant operations.
        \item [iii.] $MNO_<(\mathscr{N},K) :=$ the set of kinetic systems generated by S-invariant and S-including operations.
        \item [iv.] $MNO(\mathscr{N},K) :=$ the set of kinetic systems generated by all operations.
    \end{enumerate}
\end{definition}
\begin{lemma}
\label{inv:dependence}
    Let $(\mathscr{N},K)$ and $(\mathscr{N}',K')$ be kinetic systems with the same positive equilibria set. Then $(\mathscr{N}',K')$ is positive dependent if and only if $(\mathscr{N},K)$ is positive dependent.
\end{lemma}

\begin{proof}
    After Proposition \ref{prop:posKerN*}, a network $\mathscr{N}$ is positive dependent $\Longleftrightarrow$ for any kinetic system $(\mathscr{N},K)$, there are rate constants such that $(\mathscr{N},K)$ has a positive equilibrium $x^*$, i.e., the coordinates of $K(x^*)$ provide a positive dependent relation for the reaction vectors of $\mathscr{N}$. Since $x^*$ is also an equilibrium of $(\mathscr{N}',K')$, the $K'(x^*)$ coordinates provide a positive dependence for the reaction vectors of $\mathscr{N}'$.
\end{proof}
Furthermore, by definition, we have the following proposition:
\begin{proposition}
    Let $(\mathscr{N},K)$ be a kinetic system. Then 
    \begin{enumerate}
        \item [i.] $MNO_=(\mathscr{N},K) \subset TRF_=(\mathscr{N},K)$,
        \item [ii.] $MNO_>(\mathscr{N},K) \subset TRF_>(\mathscr{N},K)$, and 
        \item [iii.] $MNO_<(\mathscr{N},K) \subset TRF_<(\mathscr{N},K)$.
    \end{enumerate}
\end{proposition}
\begin{remark}
    In general, $MNO(\mathscr{N},K)$ may not be contained in $TRF(\mathscr{N},K)$ since the combination of extending and including operations may not result in a network transformation.
\end{remark}
\subsection{An application to the study of Wnt signaling}
In this Section, we cite an initial application of the “splitting by reaction vector” by Hernandez et al. \cite{HELM2024} in their compararative study of Wnt signaling models. An application of the “scaling” operation is discussed in Section \ref{7.2}. Hernandez et al. compared 3 multi-stationary models of Wnt signaling which fit experimental data well but are not dynamically equivalent. As part of their Common Embedded Networks (CSEN) analysis, they introduced the concept of “proximate equivalence” of two kinetic systems, which denotes the dynamic equivalence of the non-flow reaction subnetworks of the two systems. They decomposed the embedded networks (relative to the subset of common species) and then decomposed them into the subnetwork of common reactions and the complementary subnetworks. For some model pairs, proximate equivalences could be constructed for the complementary subnetworks via network transformation consisting of reaction vector splittings, signifying greater similarity.

\section{Network transformations and kinetic system properties}
\label{Sec:7}

\subsection{A PL-NDK system for a model of biosynthesis in \textit{Populus xylem}}
\label{ex2}
We illustrate the variation of structural properties using the generalized mass action (GMA) system of monolignol biosynthesis in \cite{LEVO2010}. A kinetic realization of this system was constructed by Arceo et al. \cite{AJLM2017}. We specifically study a weakly reversible transform of its maximal positive dependent subnetwork, which we denote by MOLIB. The CRN of the kinetic realization (denoted as ECJ5-G in \cite{AJLM2017} is given by

 \begin{equation}
  \begin{split}
R_{1}: & X_{1}+X_{13} \rightarrow  2X_{1}+X_{13} \\
R_{2}: & X_{1} \rightarrow X_{2} \\ 
R_{3}: & X_{2}+ X_{5} \rightarrow  X_{3}+ X_{5}  \\
R_{4}: & X_{2} \rightarrow X_{5}  \\
R_{5}: & X_{2 } \rightarrow 0  \\
R_{6}: & X_{3} \rightarrow  X_{4}  \\
R_{7}: & X_{3} \rightarrow X_{6}  \\
R_{8}: & X_{3} \rightarrow 0  \\
R_{9}: & X_{4} \rightarrow 0  \\
R_{10}: & X_{5} \rightarrow  X_{6}  \\
R_{11}: & X_{5}+X_{11} \rightarrow  X_{7}+ X_{11}  \\
R_{12}: & X_{5}+X_{7} \rightarrow  X_{5}+ X_{8}  \\
R_{13}: & X_{6} \rightarrow  X_{8} 
  \end{split}
\quad \quad 
  \begin{split}
R_{14}: & X_{6}+X_{8} \rightarrow  X_{6}+ X_{9}  \\
R_{15}: & X_{7}+X_{9} \rightarrow  X_{9}+ X_{10}  \\
R_{16}: & X_{9} \rightarrow 0  \\
R_{17}: & X_{9} \rightarrow X_{11}  \\
R_{18}: & X_{10}+X_{11} \rightarrow  X_{11}  \\
R_{19}: & X_{11} \rightarrow 0   \\
R_{20}: & X_{11} \rightarrow X_{12}  \\
R_{21}: & X_{12} \rightarrow 0  \\
R_{22}: & 0 \rightarrow X_{13}  \\
R_{23}: & X_{4} \rightarrow H  \\
R_{24}: & X_{9} \rightarrow G  \\
R_{25}: & X_{12} \rightarrow S \\
  \end{split}
  \nonumber
\end{equation}

It is not positive dependent according to CRNToolbox so we obtain its maximal positive dependent subnetwork by removing the last 4 reactions and denote it by ECJ5-GMPD. The kinetic order matrix for ECJ5-GMPD is given by:
\includegraphics[width=\textwidth]{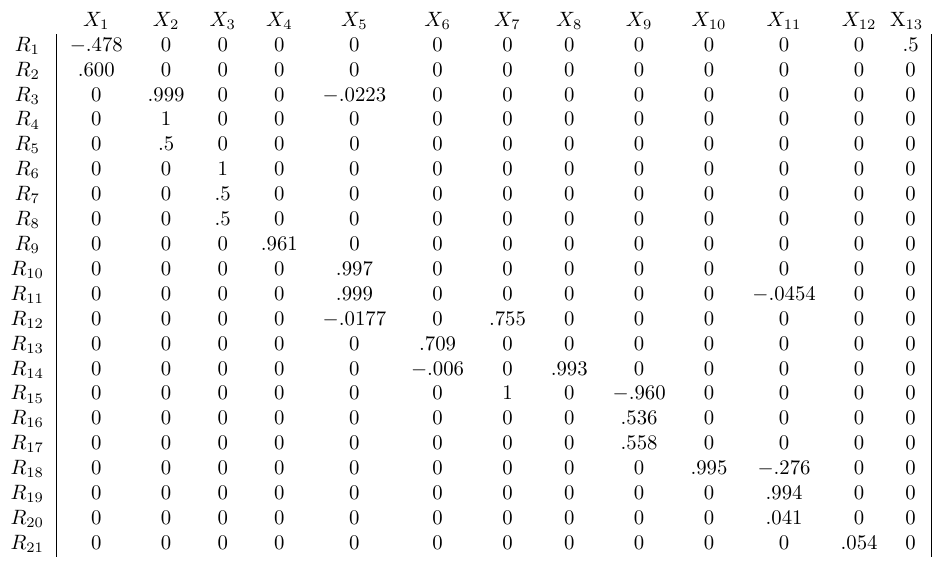}
As noted in [AJLM2017], it is a PL-NDK system with the following NDK nodes as shown in Table \ref{tabNDK}.
\begin{table}[h!]
\caption{NDK properties of ECJ5-GMPD}\label{tabNDK}%
\begin{tabular}{lcc}
\hline
Branching NDK complex & Reaction numbers of branching reactions & Comment \\
\hline
$X_2 + X_{13}$ & $R_3, R_4, R_5$ &  \\
$X_3 + X_{13}$ & $R_6, R_7, R_8$ & 
 $R_7, R_8$ have the same kinetic vectors  \\
$X_9 + X_{13}$ & $R_{16}, R_{17}$ &  \\
$X_{11} + X_{13}$ & $R_{19}, R_{20}$ &  \\
\hline
\hline
\end{tabular}
\end{table}
Using only the shifting operation, we obtain a weakly reversible transform MOLIB. 
\begin{equation}
\nonumber
\begin{split}
R_{1}: & X_{13} \rightarrow  X_{1}+X_{13} \\
R_{2}: & X_{1}+X_{13 } \rightarrow X_{2 }+X_{13} \\
R_{3}: & X_{2}+X_{13 } \rightarrow  X_{3}+X_{13} \\
R_{4}: & X_{2}+X_{13 } \rightarrow X_{5}+X_{13} \\ 
R_{5}: & X_{2}+X_{13 } \rightarrow  X_{13} \\ 
R_{6}: & X_{3}+X_{13 } \rightarrow  X_{4}+X_{13} \\ 
R_{7}: & X_{3}+X_{13 } \rightarrow X_{6 }+X_{13} \\
R_{8}: & X_{3}+X_{13 } \rightarrow X_{13} \\ 
R_{9}: & X_{4}+X_{13 } \rightarrow X_{13} \\ 
R_{10}: & X_{5}+X_{13 } \rightarrow  X_{6}+X_{13} \\ 
R_{11}: & X_{5}+X_{13} \rightarrow  X_{7}+ X_{13}
\end{split} \quad \quad
\begin{split}
R_{12}: & X_{7}+X_{13} \rightarrow  X_{8 }+X_{13} \\
R_{13}: & X_{6}+X_{13} \rightarrow  X_{8}+X_{13} \\ 
R_{14}: & X_{8}+X_{13 } \rightarrow  X_{9}+X_{13} \\ 
R_{15}: & X_{7}+X_{13 } \rightarrow  X_{10}+X_{13} \\
R_{16}: & X_{9}+ X_{13} \rightarrow  X_{13} \\
R_{17}: & X_{9}+ X_{13} \rightarrow X_{11}+ X_{13} \\
R_{18}: & X_{10}+X_{13} \rightarrow  X_{13} \\
R_{19}: & X_{11}+X_{13} \rightarrow  X_{13} \\ 
R_{20}: & X_{11}+X_{13} \rightarrow X_{12}+X_{13} \\
R_{21}: & X_{12}+X_{13} \rightarrow  X_{13} 
\end{split}
\end{equation}

In this translation, we shift by adding $X_{13}$ to reactions $R_{2}$, $R_{4}$, $R_{5}$, $R_{6}$, $R_{7}$, $R_{8}$, $R_{9}$, $R_{10}$, $R_{13}$, $R_{16}$, $R_{17}$, $R_{19}$, $R_{20}$ and $R_{21}$. We also shift by adding $-X_1$ to $R_1$, $X_{13}-X_5$ to $R_3$ and $R_{12}$, $X_{13}-X_{11}$ to $R_{11}$ and $R_{18}$, $X_{13}-X_{6}$ to $R_{14}$ and $X_{13}-X_{9}$ to $R_{15}$. Table \ref{netprop:MOLIB_ECJ} shows the variation in network numbers of the two kinetic systems.

\begin{table}[h!]
\centering
\caption{Network numbers of ECJ5-GMPD and MOLIB}\label{netprop:MOLIB_ECJ}%
\begin{tabular}{lcc}
\hline
Network numbers & ECJ5-GMPD & MOLIB \\
\hline
Number of species ($m$) & 13 & 13 \\
Number of complexes ($n$) & 24 & 13 \\
Number of irreversible reactions ($rirr$) & 21 & 21 \\
Number of reversible reactions ($rrev$) & 0 & 0 \\
Number of reactions ($r$) & 21 & 21 \\
Number of linkage classes ($l$) & 7 & 1 \\
Number of strong linkage classes ($sl$) & 24 & 1 \\
Number of terminal strong linkage classes ($t$) & 8 & 1 \\
Rank of network ($s$) & 12 & 12 \\
Deficiency ($\delta$) & 5 & 0 \\
\hline
\hline
\end{tabular}
\end{table}

The network transformation has a remarkable simplifying effect: the numbers of complexes and of linkage classes are both reduced, while the number of reactions and rank are maintained, resulting in a weakly reversible deficiency zero system. Table \ref{netprop:MOLIB_ECJ2} summarizes the difference in some basic network properties of the two systems.

\begin{table}[h!]
\centering
\caption{Network properties of ECJ5-GMPD and MOLIB}\label{netprop:MOLIB_ECJ2}%
\begin{tabular}{lcc}
\hline
Basic network properties & ECJ5-GMPD & MOLIB \\
\hline
ILC & NO & YES \\
Weakly Reversible & NO & YES \\
t-minimal & NO & YES \\
Positive Dependence & YES & YES \\
Endotactic & NO & YES \\
\hline
\hline
\end{tabular}
\end{table}

\subsection{Network concordance and network transformations}
\label{7.2}
Concordant networks were introduced by Shinar and Feinberg in 2012 \cite{SF2012} as an abstraction of continuous flow stirred tank reactors (CFSTRs), a widely used model in chemical engineering. In their view, concordance indicates “…architectures that by their very nature, enforce duller, more restrictive behavior despite what might be great intricacy in the interplay of many species, even independently of values that kinetic parameters might take”. Concordance can hence be seen as a new type of system stability. 

\subsubsection{Concepts and basic properties of concordance}

To precisely define concordance, consider the linear map  $L: \mathbb{R}^\mathscr{R} \rightarrow S$ defined by

$$L(\alpha) = \sum_{q: y \rightarrow y' \in \mathscr{R}} \alpha_q (y'-y).$$

\begin{definition}
    A reaction network $\mathscr{N}$ is \textbf{concordant} if there do not exist an $\alpha \in \ker (L)$ and a nonzero $\sigma \in S$ having the following properties:
\begin{enumerate}
    \item For each $y \rightarrow y'$ such that $\alpha_{ y \rightarrow y'} \neq 0$, $\supp(y)$ contains a species $A$ for which $\sgn(\sigma_A) = \sgn(\alpha_{ y \rightarrow y'})$, where $\sigma_A$ denotes the term in sigma involving the species $A$ and $\sgn(\cdot)$ is the signum function.
    \item For each $y \rightarrow y'$ such that $\alpha_{ y \rightarrow y'} = 0$, either $\sigma_A = 0$ for all $A \in \supp(y)$, or else $\supp(y)$ contains species $A$ and $A'$ for which $\sgn(\sigma_A)=\sgn(\sigma_{A'})$, but not zero.
\end{enumerate}

A network that is not concordant is \textbf{discordant}.
    
\end{definition}

\begin{example}
    The CRNToolbox shows that Example \ref{ex1} in Appendix \ref{fundamentals} is discordant. Here, we compute
$$S = 	\{ \sigma \in \mathbb{R}^m \mid \sigma_X = -\sigma_Y \} \text{ and } \ker L = \{ \alpha \in \mathbb{R}^r \mid \alpha_{R_1} = \alpha_{R_2} \}.$$

However, shifting $R_1$ by adding $-Y$ will make the translated network $X \Leftrightarrow Y$ concordant.  
\end{example}

Concordance is closely related to two classes of kinetics on a network: injective and weakly monotonic kinetics. We recall these notions from \cite{SF2012}.
 
\begin{definition}
    A kinetic system $(\mathscr{N},K)$ is \textbf{injective} if, for each pair of distinct stoicihometrically compatible vectors $x^*, x^{**} \in \mathbb{R}^\mathscr{S}_{\geq 0}$, at least one of which is positive,

$$\sum_{y \rightarrow y'} K_{y \rightarrow y'} (x^{**}) (y'-y) \neq \sum_{y \rightarrow y'} K_{y \rightarrow y'} (x^{*}) (y'-y).$$
\end{definition}
 
Note that an injective kinetic system is necessarily a monostationary system. Moreover, an injective kinetic system cannot admit two distinct stoichiometrically compatible equilibria, at least one of which is positive. 
 
\begin{definition}
    A kinetics $K$ for a reaction network $\mathscr{N}$ is \textbf{weakly monotonic} if, for each pair of vectors  $x^*, x^{**} \in \mathbb{R}^\mathscr{S}_{\geq 0}$, the following implications hold for each reaction $y \rightarrow y' \in \mathscr{R}$ such that $\supp(y) \subset \supp(x^*)$ and $\supp(y) \subset \supp(x^{**})$:
\begin{enumerate}
    \item $K_{y \rightarrow y'} (x^{**}) > K_{y \rightarrow y'} (x^{*}) \Longrightarrow$ there exists a species A $\in \supp(y)$ with $x_A^{**} > x_A^*$.
    \item $K_{y \rightarrow y'} (x^{**}) = K_{y \rightarrow y'} (x^{*}) \Longrightarrow x_A^{**} = x_A^*$ for all $A \in \supp(y)$ or else there are species $A,A' \in \supp(y)$ with  $x_A^{**} > x_A^*$ and $x_{A'}^{**} < x_{A'}^*$.
\end{enumerate}
    
\end{definition}

\begin{remark}
    Examples of weakly monotonic kinetic systems are mass action systems and a class of power law systems where all kinetic orders are non-negative called non-inhibitory kinetics (PL-NIK systems).
\end{remark}

The following results present the close relationship between concordant networks, injective and weakly monotonic kinetics:
 
\begin{proposition}[Proposition 4.8 of \cite{SF2012}] A weakly monotonic kinetic system $(\mathscr{N} ,K)$ is injective whenever its underlying reaction network is concordant. In particular, if the underlying reaction network is concordant, then the kinetic system cannot admit two distinct stoichiometrically compatible equilibria, at least one of which is positive.
\end{proposition} 
 
\begin{theorem} [Theorem 4.11 of \cite{SF2012}] A reaction network has injectivity in all weakly
monotonic kinetic systems derived from it if and only if the network is concordant.
\end{theorem} 
 
\begin{proposition}[Proposition 10.5.8 of \cite{FEIN2019}] A network $\mathscr{N}$ is concordant iff every PL-NIK kinetics on $\mathscr{N}$ is injective. 
\end{proposition} 
 
\begin{example}
    Application of CRNToolbox's Concordance Test to ECJ5-GMPD verified its concordance. It is only the 4th published biochemical system with a concordant underlying network—previously, only the olfactory calcium signaling system \cite{SF2012}, Lee's Wnt signaling model \cite{FEIN2019} and the metabolic insulin signaling system of Sedaghat et al. \cite{SEDA2002} were known to have the property. Furthermore, ECJ5-GMPD is the first biochemical example with non-mass action kinetics.
\end{example}

\subsubsection{A sufficient condition for the invariance of concordance under network transformations}
We conclude our discussion of network property variation under network transformation with several new results on concordance. We first define a special classes of shifting and dividing operations: 
\begin{definition}
    A shift of a reaction $y \rightarrow y'$ to $y+y^\sharp
 \rightarrow y'+y^\sharp$ is called \textbf{reactant support preserving (rsp)} if $\supp (y + y^\sharp)  =  \supp y$. A dividing operation is rsp if both associated shifts are rsp.
\end{definition}

Since $\supp y \subset \supp (y + y^\sharp)$ always holds, the property $supp y = supp(y + y^\sharp)$ is equivalent to $\supp (y + y^\sharp) \subset  \supp y$. As one can easily compute, this can be further simplified and we have:

\begin{proposition}
    For any shifting operation, $\supp (y + y^\sharp) \subset \supp y$ if and only if $\supp (y^\sharp) \subset \supp y$.    
\end{proposition}

The proofs are straightforward computations. The next Proposition identifies two building blocks of network transformations that conserve concordance:

\begin{proposition}
    Let $(\mathscr{N}^\sharp, K^\sharp)$ be a network transform of $(\mathscr{N}, K)$. Let  the transformation consists only of

    \begin{enumerate}
        \item[a.] rsp-shifting,
        \item[b.] rsp-dividing or
        \item[c.] scaling
    \end{enumerate}

Then 
    \begin{enumerate}
        \item[i.] $r \leq r^\sharp$ and
        \item[ii.] if $(\mathscr{N}, K)$ is concordant, them $(\mathscr{N}^\sharp, K^\sharp)$ is also concordant
    \end{enumerate}
\end{proposition}

\begin{proof}
For (i), both the shifting and scaling operations do not change the number of reactions, while each dividing operation increases it by 1. For (ii) the assumption $\supp y^\sharp \subset \supp y$ implies that $\supp (y + y^\sharp) = \supp y$ so that $K^\sharp$ remains PL-NIK. For a scaled SFRF, we have $f^\sharp = \lambda N \lambda^{-1} (\lambda K)$ which is injective on the same stoichiometric class as $K$. Hence, any PL-NIK system is injective on N if and only if its transform is injective on $\mathscr{N}^\sharp$.
\end{proof}

\begin{example}
    Application of the CRNToolbox test shows that the CF-WR translate of the Schmitz carbon cycle model is discordant. This is  because some of the shifting operators do not satisfy the additional property of $\supp y^\sharp \subset \supp y$, highlighting the essentiality of the property.
\end{example}

\begin{erratum}
The previous example also shows that statement ii) in Theorem 9 in \cite{FOME2023} does not hold in general. However, the CRNToolbox test confirms that the linear conjugate studied is concordant. 
\end{erratum}

\begin{example}
    For the Schmitz system, the motto is “Scale, don't shift: by scaling reactions $R_1, R_2, R_4, R_7$ by $\lambda = 2$ and obtaining $2M_3 \Leftrightarrow 2M_1 \Leftrightarrow 2M_2$ and leaving the rest unchanged, we obtain a weakly reversible PL-RDK network transform which is concordant.  Note that this is not a network translation, showing an advantage of the bigger set of network operations.
\end{example}

\subsection{Structo-kinetic properties and network transformations}

Beyond structural properties, a kinetic system has two classes of properties: kinetic and structo-kinetic properties. Kinetic properties are those that are invariant under dynamic equivalence, e.g. the set of positive equilibria. The most studied structo-kinetic property is the set of complex-balanced equilibria of a kinetic system - its variation under network transformations is apparent from Horn's classical necessary condition for complex balancing, i.e., the existence of complex balanced equilibria, namely weak reversibility of the underlying network. \\

The three biochemical systems we have previously analyzed are all not weakly reversible, and hence after Horn's result, not complex balanced. MOLIB however, again after a classical result, this time from Feinberg, is complex balanced, since it has zero deficiency. In fact, it is absolutely complex balanced, i.e., every positive equilibrium is a complex balanced one. \\ 

After a brief assessment of the invariance of injectivity and mono-/multi-stationarity under network transformations, we present the main result of this section: an algorithm for the transformation of a weakly reversible NFK system to a weakly reversible CFK system. We conclude the section by illustrating this result using MOLIB.

\subsubsection{Variation of injectivity and mono-/multi-stationarity of kinetic systems}
For any differentiable kinetic system, it is well known that the trajectory of a solution remains in the stoichiometric class containing the vector of initial conditions. Hence, answers to questions about the number of its equilibria are relative to the corresponding stoichiometric class. This clearly implies that the properties of injectivity and mono-/multi-stationarity of a kinetic system are structo-kinetic properties since a dynamic equivalence does not necessarily conserve the stoichiometric classes. \\

The following Proposition summarizes the relationships of these properties for some classes of network transformations:

\begin{proposition} \label{prop:ten}
Let $(\mathscr{N}^\#,K^\#)$ be a transform of $(\mathscr{N},K)$. Then
\begin{enumerate}
    \item If the transformation is S-including, then $(\mathscr{N},K)$ injective $\Rightarrow$ $(\mathscr{N}^\#,K^\#)$ injective and $(\mathscr{N},K)$  mono-stationary $\Rightarrow$ $(\mathscr{N}^\#,K^\#)$  mono-stationary. \label{una}

    \item If the transformation is S-extending, then $(\mathscr{N},K)$ non-injective $\Rightarrow$ $(\mathscr{N}^\#,K^\#)$ non-injective and $(\mathscr{N},K)$  multi-stationary $\Rightarrow$ $(\mathscr{N}^\#,K^\#)$  multi-stationary.

    \item If the transformation is S-invariant, then all statements hold. \label{tres}
\end{enumerate}
\end{proposition}

\begin{proof}
    For (\ref{una}): $(\mathscr{N}^\#,K^\#)$ non-injective $\Rightarrow$ there exist $x,x'$ with $x'-x$ in $S^\#$ and $N^\# K^\# (x) = N^\# K^\# (x')$. Since $N^\# K^\# = NK$ and $S^\#$ is contained in $S$, this contradicts the injectivity of $(\mathscr{N}, K)$. The arguments for the other implications are similar. For (\ref{tres}): It follows from the fact that S-invariant transformations are both S-including and S-extending.  
\end{proof}

\begin{example}
    Since the olfactory calcium signaling system is concordant and has mass action kinetics, it follows from Proposition \ref{prop:ten} that the system is injective. Since the transformation in Section \ref{4.2.3} is S-invariant, then its weakly reversible transform, though discordant, is still injective (and mono-stationary).  
\end{example}   

\subsubsection{Complex balancing, complex factorizability and network transformations}
With regards to the structo-kinetic properties of complex balancing and complex factorizability, which are useful for analyses, we have hitherto focused on finding transformations leading to them. The question of their invariance under network transformations has not been discussed. Fortunately, it is easily answered in the sense that the only operation preserving them in general is the scaling operation. The effect of all other operations is very network specific and hence not amenable to more general statements.

\section{Summary/Outlook}
\label{Sec:8}

The study is motivated by the fact that the graphical structure of the reaction network is crucial for gaining mathematical insights, such as equilibria, multiple equilibria, concordance, stability, and concentration robustness. The study of network translation by M. Johnston \cite{JOHN2014}, later extended by Hong et al. \cite{HONG2023}, introduced operations for modifying the structure of the network while simultaneously preserving the ordinary differential equations (ODEs). In this paper, we improve and extend their work. To summarize, here are the main results:

\begin{enumerate}
\item We establish that positive dependence is a necessary and sufficient condition for a weakly reversible dynamically equivalent system. That is, in a positively dependent system, there is a sequence of shifting, merging, and dividing operations that will make it weakly reversible.

\item We demonstrate the benefits of kinetic realizations with positive deficiency for the analysis of biochemical systems. Here, we provide two examples: (1) calcium signaling in olfactory cilia by Reidl et al. \cite{REIN2006}, and (2) metabolic insulin signaling in adipocytes by Sedaghat et al. \cite{SEDA2002}.

\item We introduce an algorithm for a network transformation of a weakly reversible non-CFK system to a weakly reversible CFK system. This leads to three important results: (1) a significant improvement of a result by Nazareno et al. \cite{NAZA2019} on the Low Deficiency Complement for the Schmitz carbon cycle model, (2) enhancement of the Subspace Coincidence Theorem for non-CFK systems, and (3) a sufficient condition for the log parametrization of PL-NDK systems.

\item We extend the concept of network translations to network transformations, introducing scaling, splitting via rate constant difference, splitting via reaction vector, merging via reaction vector, and adding dummy reactions as additional operations. In this part, we introduce dynamical equivalences that adhere to stoichiometric alignment properties, namely S-extending, S-including, and S-invariant transformations.

\item Lastly, using the biosynthesis of monolignol in the \textit{Populus xylem} model of Lee and Voit \cite{LEVO2010}, we discuss the variation or invariance of structural and kinetic properties under different subsets of transformations. Additionally, we establish a sufficient condition for the invariance of concordance under a network transformation in terms of its constituent network operations.
\end{enumerate}

As future perspective, one may look at additional system properties which ensure the invariance/variation of the property under network operations. Furthermore, one may integrate with the existing research on decompositions for analyzing reaction networks in  biochemical systems. This integration could provide a more comprehensive understanding of the underlying dynamics and behavior of such systems.
\section*{Acknowledgments}
The authors thank the two reviewers of an earlier version whose comment contributed greatly to improving the manuscript.

\baselineskip=0.25in
\bibliographystyle{unsrt}

\begin{thebibliography}{1}

\bibitem{FEHO1977}Feinberg, M. \& Horn, F. Chemical mechanism structure and the coincidence of the stoichiometric and kinetic subspaces.. {\em Archive For Rational Mechanics And Analysis}. \textbf{66}, 83-97 (1977)
\bibitem{WIFE2013}Wiuf, C. \& Feliu, E. Power-Law Kinetics and Determinant Criteria for the Preclusion of Multistationarity in Networks of Interacting Species. {\em SIAM Journal On Applied Dynamical Systems}. \textbf{12}, 1685-1721 (2013), 
\bibitem{HKK2021}Hong, H., Kim, J., Al-Radhawi, M., Sontag, E. \& Kim, J. Derivation of stationary distributions of biochemical reaction networks via structure transformation. {\em Commun Biol}. pp. 620 (2021)
\bibitem{FOME2020}Fortun, N. \& E.R. Mendoza Absolute concentration robustness in power law kinetic systems. {\em MATCH Communications In Mathematical And In Computer Chemistry}. pp. 669-691 (2020)
\bibitem{HELM2024}Hernandez, B., Lubenia, P. \& E.R. Mendoza  Embedding-based comparison of reaction networks of Wnt signaling. {\em (Accepted At) MATCH Communications In Mathematical And In Computer Chemistry }. (2024)
\bibitem{FOME2023}Fortun, N. \& E.R. Mendoza Comparative analysis of carbon cycle models via kinetic representations. {\em Journal Of Mathematical Chemistry}. pp. 896-932 (2023)
\bibitem{TAME2018}Talabis, D., Mendoza, E. \& Jose, E. Positive Equilibria of Weakly Reversible Power Law Kinetic Systems with Linear Independent Interactions. {\em Journal Of Mathematical Chemistry}. \textbf{56}, 2643-2673 (2018)
\bibitem{HOJA1972}Horn, F. \& Jackson, R. General mass action kinetics. {\em Archive For Rational Mechanics And Analysis}. \textbf{47} pp. 81-116 (1972)
\bibitem{MTJ2022}Jose, E., Talabis, D. \& Mendoza, E. Absolutely Complex Balanced Kinetic Systems. {\em MATCH Commun. Math. Comput. Chem.}. \textbf{88} pp. 397-436 (2022)
\bibitem{JOSS2013}Johnston, M., Siegel, D. \& Szederkényi, G. Computing weakly reversible linearly conjugate chemical reaction networks with minimal deficiency. {\em Mathematical Biosciences}. \textbf{241}, 88-98 (2013), https://www.sciencedirect.com/science/article/pii/S0025556412001952
\bibitem{MHNN2020}Magpantay, D., Hernandez, B., Delos Reyes, A., Mendoza, E. \& Nocon, E. A computational approach to multistationarity in poly-PL kinetic systems. {\em MATCH Commun. Math. Comput. Chem.}. \textbf{85} pp. 605-634 (2021)
\bibitem{LLMM2022}Lao, A., Luvenia, P., Magpantay, D. \& Mendoza, E. Concentration Robustness in LP Kinetic Systems. {\em MATCH Commun. Math. Comput. Chem.}. \textbf{88} pp. 29-66 (2022)
\bibitem{FLRM2019}Fortun, N., Lao, A., Razon \& Mendoza, E. A Deficiency Zero Theorem for a class of power law kinetic systems with non-reactant determined interactions. {\em MATCH Commun. Math. Comput. Chem.}. \textbf{81} pp. 621-638 (2019)
\bibitem{SF2010}Shinar, G. \& Martin Feinberg Structural Sources of Robustness in Biochemical Reaction Networks. {\em Science}. \textbf{327}, 1389-1391 (2010), https://www.science.org/doi/abs/10.1126/science.1183372
\bibitem{FLRM2021}Fortun, N., Lao, A., Razon, L. \& Mendoza, E. Robustness in Power-Law Kinetic Systems with Reactant-Determined Interactions. {\em Discrete And Computational Geometry, Graphs, And Games}. pp. 106-121 (2021)
\bibitem{LUML2022}Lubenia, P., Mendoza, E. \& Lao, A. Reaction Network Analysis of Metabolic Insulin Signaling. {\em Bull. Math. Biol.}. \textbf{84} pp. 129 (2022)
\bibitem{SF2012}Shinar, G. \& Feinberg, M. Concordant chemical reaction networks. {\em Mathematical Biosciences}. \textbf{240}, 92-113 (2012), https://www.sciencedirect.com/science/article/pii/S0025556412001204
\bibitem{LMF2021}Fontanil, L., Mendoza, E. \& Fotun, N. A computational approach to concentration robustness in power law kinetic systems of Shinar-Feinberg type. {\em MATCH Commun. Math. Comput. Chem.}. \textbf{85} pp. 489-516 (2021)
\bibitem{FEME2021}Fotun, N. \& Mendoza, E. Absolute Concentration Robustness in Power Law Kinetic Systems. {\em MATCH Commun. Math. Comput. Chem.}. \textbf{85} pp. 669-691 (2021)
\bibitem{HEME2021}Hernandez, B. \& Mendoza, E. Positive equilibria of Hill-type kinetic systems. {\em J. Math. Chem.}. \textbf{59} pp. 840-870 (2021)
\bibitem{MURE2014}Müller, S. \& Regensburger, G. Generalized Mass-Action Systems and Positive Solutions of Polynomial Equations with Real and Symbolic Exponents (Invited Talk). {\em Computer Algebra In Scientific Computing}. pp. 302-323 (2014)
\bibitem{HEME2023}Hernandez, B. \& Mendoza, E. Positive equilibria of power law kinetics on networks with independent linkage classes. {\em J. Math. Chem.}. \textbf{61} pp. 630-651 (2023)
\bibitem{JOHN2014}Johnston, M. Translated Chemical Reaction Networks. {\em Bull. Math. Biol.}. \textbf{76}, 1081-1116 (2014)
\bibitem{TONJOHN2018}Tonello, E. \& Johnston, M. Network translation and steady state properties of chemical reaction systems. {\em Bull. Math. Biol.}. \textbf{80}, 2306-2337 (2018)
\bibitem{CRACIUN2009}Craciun, G., Dickenstein, A., Shiu, A. \& Sturmfels, B. Toric dynamical systems. {\em Journal Of Symbolic Computation}. \textbf{44}, 1551-1565 (2009), https://www.sciencedirect.com/science/article/pii/S0747717109000923, In Memoriam Karin Gatermann
\bibitem{JOHNMUPA2018}Johnston, M., Müller, S. \& Pantea, C. A Deficiency-Based Approach to Parametrizing Positive Equilibria of Biochemical Reaction Systems. {\em Bull. Math. Biol.}. \textbf{81} pp. 1143-1172 (2019)
\bibitem{BOROSYU2019}Boros, B., Craciun, G. \& Yu, P. Weakly Reversible Mass-Action Systems With Infinitely Many Positive Steady States. {\em SIAM Journal On Applied Mathematics}. \textbf{80}, 1936-1946 (2020)
\bibitem{MESH2022}Meshkat, N., Shiu, A. \& Torres, A.  Absolute Concentration Robustness in Networks with Low-Dimensional Stoichiometric Subspace. {\em Vietnam J. Math.}. \textbf{50} pp. 623-651 (2022)
\bibitem{HONG2023}Hong, H., Hernandez, B., Kim, J. \& Kim, J. Computational Translation Framework Identifies Biochemical Reaction Networks with Special Topologies and Their Long-Term Dynamics. {\em SIAM Journal On Applied Mathematics}. \textbf{83}, 1025-1048 (2023)
\bibitem{AJLM2017}Arceo, C., Jose, E., Lao, A. \& Mendoza, E. Reaction networks and kinetics of biochemical systems. {\em Mathematical Biosciences}. \textbf{283} pp. 13-29 (2017), https://www.sciencedirect.com/science/article/pii/S0025556416302486
\bibitem{FEIN1987}Feinberg, M. Chemical reaction network structure and the stability of complex isothermal reactors—I. The deficiency zero and deficiency one theorems. {\em Chemical Engineering Science}. \textbf{42}, 2229-2268 (1987), https://www.sciencedirect.com/science/article/pii/0009250987800994
\bibitem{LEVO2010}Lee, Y. \& Voit, E. Mathematical modeling of monolignol biosynthesis in Populus xylem. {\em Mathematical Biosciences}. \textbf{228}, 78-89 (2010), https://www.sciencedirect.com/science/article/pii/S0025556410001410
\bibitem{NAZA2019}Nazareno, A., Eclarin, R., Mendoza, E. \& Lao, A. Linear conjugacy of chemical kinetic systems. {\em Math. Biosci. Eng.}. \textbf{16}, 8322-8355 (2019)
\bibitem{REIN2006}Reidl, J., Borowski, P., Sensse, A., Starke, J., Zapotocky, M. \& Eiswirth, M. Model of Calcium Oscillations Due to Negative Feedback in Olfactory Cilia. {\em Biophysical Journal}. \textbf{90}, 1147-1155 (2006), https://www.sciencedirect.com/science/article/pii/S0006349506723062
\bibitem{SEDA2002}Sedaghat, A., Sherman, A. \& Quon, M. A mathematical model of metabolic insulin signaling pathways. {\em American Journal Of Physiology-Endocrinology And Metabolism}. \textbf{283}, E1084-E1101 (2002), PMID: 12376338
\bibitem{MURE2012}Müller, S. \& Regensburger, G. Generalized Mass Action Systems: Complex Balancing Equilibria and Sign Vectors of the Stoichiometric and Kinetic-Order Subspaces. {\em SIAM Journal On Applied Mathematics}. \textbf{72}, 1926-1947 (2012)


\end{thebibliography}

\appendix
\section{Fundamentals}
\label{fundamentals}
A \textbf{reaction network} is a system of unique reactions. The structure of a network can be viewed as a directed graph where the edges are the \textbf{reactions} and the \textbf{complexes} (i.e., vertices) are non-negative linear combinations of the \textbf{species} of the network. The set of species, set of complexes and set of reactions are denoted as $\mathscr{S}$, $\mathscr{C}$, and $\mathscr{R}$, respectively. Furthermore, we denote $m = |\mathscr{S}|$, $n =|\mathscr{C}|$ and $r = |\mathscr{R}|$. We associate each reaction with the difference between its product and source nodes called a \textbf{reaction vector}.

\begin{example}
The following example is a reaction network with two reactions, two species and four complexes.
\begin{center}
    \begin{tikzpicture}[baseline=(current  bounding  box.center)]
    \tikzset{vertex/.style = {draw=none,fill=none}}
    \tikzset{edge/.style = {bend left,->,> = latex', line width=0.20mm}}
    \node[vertex] (1) at  (1,0) {$X + Y$};
    \node[vertex] (2) at  (4,0) {$2Y$};
    \node[vertex] (3) at  (1,-1) {$X$};
    \node[vertex] (4) at  (4,-1) {$Y$};
    \draw [edge]  (1) to["$r_1$"] (2);
    \draw [edge]  (4) to["$r_2$"] (3);
    \end{tikzpicture} 
\end{center}

In this example, we have $\mathscr{S}=\left\{ X, Y\right\}$,  $\mathscr{C}=\left\{ X+Y, 2Y, Y, X \right\}$ and  $\mathscr{R}=\left\{ R_1: X+Y \rightarrow 2Y, R_2: Y \rightarrow X\right\}$. For $R_1: X+Y \rightarrow 2Y$, the complexes $X+Y$ and $2Y$ are called the \textbf{reactant complex} and the \textbf{product complex}, respectively. Also, the reaction vectors for $R_1$ and $R_2$ are given as follows:

$$\begin{array}{l}
R_1: 2Y - (X+Y) = -X + Y = [-1 \quad 1]^\top \\
R_2: X - Y = [1 \quad -1]^\top  \\
\end{array}.$$ 
\label{ex1}
\end{example}

The \textbf{stoichiometric matrix} $N$ of a CRN is a matrix whose columns are the reaction vectors. The \textbf{stoichiometric subspace} $S$ is the image of $N$, which is equivalent to the linear subspace of $\mathbb{R}^m$ defined by:
\begin{equation*}
    S := \text{span } \left\lbrace y' - y \in \mathbb{R}^m \mid y\rightarrow y' \in \mathscr{R}\right\rbrace.
\end{equation*}
 \noindent The dimension of $S$, denoted by $s$, gives the \textbf{rank} of the CRN. We define an important concept in CRN theory called the \textbf{deficiency} of a CRN, denoted by $\delta$, which can be easily calculated using the formula $\delta = n - l -s$, where $n$ is the number of complexes, $l$ is the number of connected components (i.e, linkage classes), and $s$ is the rank of the CRN. Table \ref{structure} summarizes the other studied structural properties in this paper.

\begin{table}[h!]
\centering
\caption{Some structural properties of kinetic systems}\label{structure}%
\begin{tabularx}{\linewidth}{XX}
\hline
Property & Definition \\
\hline
\hline
reversible reaction & reaction that can occur in both the forward and reverse directions \\
linkage classes & connected components of a CRN \\
strong linkage classes & strongly connected components of a CRN \\
terminal strong linkage classes & strongly connected components without outgoing arcs \\
independent linkage classes (ILC) & the total deficiency of the linkage classes is equal to the network deficiency \\
weakly reversible & every linkage class is a strong linkage class \\
t-minimal & every linkage class has one terminal strong linkage class \\
\hline
\hline
\end{tabularx}
\end{table}

In Example \ref{ex1}, we have $N=\begin{blockarray}{ccc}
R_1 & R_2  \\
\begin{block}{[cc]c}
-1 & 1 & X \\ 
1 & -1 & Y \\ 
\end{block}
\end{blockarray}$ and $s=1$. Furthermore, the network has two linkage classes, not weakly reversible but t-minimal. \\

The \textbf{matrix of complexes} $Y$ is the $m \times n$
 matrix whose ($i,j$)th element is the stoichiometric coefficient of $y_j$
 with respect to $C_i$. The \textbf{incidence matrix} 
 $I_a$ is $n \times r$  matrix where each row corresponds to a complex and each column to a reaction, satisfying

 \begin{equation*}
     \nonumber
    (I_a)_{i,j}=\begin{cases}
    -1,& \text{if } i \text{ is the reactant complex of reaction } j \in \mathscr{R},\\
    1,              & \text{if } i \text{ is the product complex of reaction } j \in \mathscr{R},\\
    0, & \text{otherwise.} 
\end{cases}
 \end{equation*}
Here, we note that the stoichiometric matrix $N$ is also defined as $N = Y I_a$. 
In Example \ref{ex1}, $\delta = 1 = 4 - 2 -1$ and the matrices $Y$ and $I_a$ are

$$Y=\begin{blockarray}{ccccc}
C_1 & C_2 & C_3 & C_4 \\
\begin{block}{[cccc]c}
1 & 0 & 0 & 1 & X \\ 
1 & 2 & 1 & 0 & Y \\ 
\end{block}
\end{blockarray} \quad I_a=\begin{blockarray}{ccc}
R_1 & R_2 \\
\begin{block}{[cc]c}
-1 & 0 & C_1 \\
1 & 0 & C_2 \\
0 & -1 & C_3 \\
0 & 1 & C_4 \\
\end{block}
\end{blockarray}.$$

A reaction network is usually endowed with a kinetics to describe how the concentrations of different species change over time.

\begin{definition} \cite{WIFE2013}
    A \textbf{kinetics} $K$ on a CRN $\mathscr{N}$ assigns to each reaction $q: y \rightarrow y'$ a rate function $K_q: \Omega \rightarrow \mathbb{R}_{\geq}$, where $\Omega$ is a subset of $\mathbb{R}^m_{\geq}$ containing $\mathbb{R}^m_{>}$, and $\supp x$ contains $\supp y$ $\Rightarrow$ $K_q(x) > 0$.
\end{definition}
 Thus, a kinetics is an assignment of a rate function to each reaction in a network. A network $\mathscr N$ together with a kinetics $K$ is called a \textbf{kinetic system} and is denoted here by $(\mathscr N, K)$. In Example \ref{ex1}, suppose that $x$ and $y$ are the concentrations of the species $X$ and $Y$, which evolve over time. Then, an example of rate functions can be as follows:

$$\begin{array}{ll}
 R_1: X+Y \rightarrow 2Y  & \quad K_1 = r_1 x^1 y^1 \\
 R_2: Y \rightarrow X & \quad K_2 = r_2 x^0 y^1 \\
\end{array}$$

where $r_i$ is the rate constant of the reaction $R_i$. Here, the kinetics $K = [r_1 x^1 y^1 \quad r_2 x^0 y^1]^\top$ is under \textbf{mass-action} (MAK) where the rate function of each reaction is proportional to the  the product of each concentration raised to the stoichiometric coefficient of the species that occurs in the reactant complex of the associated reaction. If we allow the exponent to be any real number, the kinetics follows \textbf{power-law kinetics} (PLK). A generalization of these systems is the RIDK system.

\begin{definition} \cite{NAZA2019}
A \textbf{rate constant-interaction map decomposable kinetics} (RIDK) is a kinetics, such that for each reaction $R_j$, the coordinate function $K_j: \Omega \rightarrow \mathbb{R}$ can be written in the form $K_j(x)=k_j I_{K,j} (x)$, with $k_j \in \mathbb{R}_{>0}$ (called rate constant) and $\Omega \subset \mathbb{R}^m$. We call the map $I_K : \Omega \rightarrow \mathbb{R}^r$ defined by $I_{K,j}$ as the \textbf{interaction map}. In these systems, each rate function is the product of a rate constant and an interaction function.
\end{definition}

Table \ref{kinetics} displays the additional kinetic systems discussed in the paper. 

\begin{table}[h!]
\centering
\caption{Some kinetics}\label{kinetics}%
\begin{tabularx}{\linewidth}{XX}
\hline
Kinetics & Definition \\
\hline
\hline
MAK & $K_i(x)=k_ix^{Y_{i,\cdot}}$ for $i=1,\dots, r$ where $Y$ is the map of complexes. \\
PLK & $K_i(x)=k_ix^{F_{i,\cdot}}$ for $i=1,\dots, r$ where $F$ is the kinetic order matrix. \\
poly-PL & $K_i(x)= \sum k_ix^{F_{i,\cdot}}$ for $i=1,\dots, r$ where $F$'s are the kinetic order matrices. \\
CFK & RIDK where for any two branching reactions $R_i, R_j\in \mathscr R$, the corresponding interaction functions are identical. \\ 
NFK & non-CFK  \\
PL-NIK  & non-inhibitory kinetics , i.e., those whose kinetic orders are all non-negative, constitute the weakly monotonic subset of power law kinetics.\\
PL-FSK & Power Law Factor Span Surjective Kinetics, i.e, subset of PL-RDK and consists of systems with the property that reactions with different reactant complexes have different kinetic order matrix rows. \\
\hline
\hline
\end{tabularx}
\end{table}

\begin{definition} 
\label{ODE:KSSC}
The \textbf{ODE} or \textbf{dynamical system} of the kinetic system is the equation $dc/dt=f(c)=NK(c)=YI_aK(c)$, where $c\in \mathbb R^{\mathscr S}_{\geq 0}$, $N$ is the stoichiometric matrix and K is the kinetics. An element $c^*$ of $\mathbb R^{\mathscr S}_{>0}$ such that $f(c^*)=0$ is called a \textbf{steady state} of the system. We use $E_+(\mathscr N, K)$ to denote the set of all positive equilibria of the kinetic system $(\mathscr N, K)$. Furthermore, we define \textbf{KSS coincidence} when the kinetic subspace ($\spa(\Ima f)$) coincides with the stoichiometric subspace $S$. Here, KSS coincidence allows the occurrence of non-degenerate steady states for differentiable kinetics.
\end{definition}

The dynamical system $f(x)$ (or species formation rate functions (SFRF)) of Example \ref{ex1} can be written as
$$\left[ 
\begin{array}{c}
    \dot{x} \\
    \dot{y} \\
\end{array}
\right]=
\left[ \begin{array}{cc}
-1 & 1  \\ 
1 & -1  \\ 
\end{array}
\right]
\left[ 
\begin{array}{c}
	r_1 x^1 y^1 \\
	r_2 x^0 y^1 \\
\end{array}
 \right]
=NK(x) \Rightarrow 
\left\{
\begin{split}
\dot{x} = - r_1 x^1 y^1 + r_2 x^0 y^1 \\
\dot{y} = r_1 x^1 y^1 - r_2 x^0 y^1. 
\end{split}\right. $$

In \cite{HOJA1972}, Horn and Jackson introduced a subset of $E_+$ called the set of complex balanced of equilibria denoted as $Z_+$. It is defined as $Z_+(\mathscr{N},K)=\left\{x\in \mathbb{R}^{\mathscr{S}}_+ \middle| I_a K(x)=0 \right\}$. A kinetic system is \textbf{complex balanced }at a state (i.e., a species composition) if for each complex, formation and degradation are at equilibrium. Table \ref{structo-kin} provides an overview of the remaining structo-kinetic properties discussed in the paper.

\begin{table}[h!]
\centering
\caption{Some structo-kinetic properties of kinetic systems}\label{structo-kin}%
\begin{tabularx}{\linewidth}{XX}
\hline
Property & Definition \\
\hline
\hline
positively equilibrated  & $E_+(\mathscr{N},K)\neq \emptyset$. \\
complex balanced  & $Z_+(\mathscr{N},K)\neq \emptyset$. \\
absolutely complex balanced (ACB) & $Z_+(\mathscr{N},K)\neq \emptyset \text{ and } E_+(\mathscr{N},K)=Z_+(\mathscr{N},K),$
i.e., all positive equilibria are complex balanced. \\
has absolute concentration robustness (ACR) in species $X$ & The concentration value for $X$ is invariant over all equilibria in $E_+(\mathscr{N},K)$.\\
has balanced concentration robustness (BCR) in species $X$ & The concentration value for $X$ is invariant over all equilibria in $Z_+(\mathscr{N},K)$. \\
multistationary & There are positive rate constants where the corresponding ODE system admits at least two distinct stoichiometrically compatible equilibria. \\
monostationary & Non-multistationary systems \\
\hline
\hline
\end{tabularx}
\end{table}

\begin{definition}
    Two kinetic systems $(\mathscr{N}, K)$ and $(\mathscr{N}', K')$ are \textbf{dynamically equivalent} if 
    \begin{enumerate}
        \item $\mathscr{N}$ and $\mathscr{N}'$ have the same set of species $S = S'$,
        \item $K$ and $K'$ have the same definition domain $\Omega = \Omega'$ and 
        \item $NK (x) = N'K'(x)$ for all $x \in \Omega$.
    \end{enumerate}
\end{definition}

An example of a dynamically equivalent system for Example \ref{ex1} is the following system $(\mathscr{N}',K')$:

$$\begin{array}{cc}
    \begin{tikzpicture}[baseline=(current  bounding  box.center)]
    \tikzset{vertex/.style = {draw=none,fill=none}}
    \tikzset{edge/.style = {bend left,->,> = latex', line width=0.20mm}}
    \node[vertex] (1) at  (1,0) {$X + Y$};
    \node[vertex] (2) at  (4,0) {$2Y$};
    \draw [edge]  (1) to["$r_1$"] (2);
    \draw [edge]  (2) to["$r_2$"] (1);
    \end{tikzpicture}  & K'(X) = \left[ 
\begin{array}{c}
	r_1 x^1 y^1 \\
	r_2 x^0 y^1 \\
\end{array}
 \right] 
\end{array}.$$

Here, the aforementioned example represents a zero deficiency weakly reversible power law system. Since $N' = \left[ 
\begin{array}{cc}
	-1 & 1 \\
	1 & -1 \\
\end{array}
 \right]$, direct computations show that $NK(X) = N'K'(X)$. We now define some kinetic system properties.\\

\begin{definition}
\label{concordance}
    Concordance can be seen as a new type of system stability. Consider the linear map  $L: \mathbb{R}^\mathscr{R} \rightarrow S$ defined by
$$L(\alpha) = \sum_{q: y \rightarrow y' \in \mathscr{R}} \alpha_q (y'-y).$$
    
    A reaction network $\mathscr{N}$ is \textbf{concordant} if there do not exist an $\alpha \in \ker (L)$ and a nonzero $\sigma \in S$ having the following properties:
\begin{enumerate}
    \item For each $y \rightarrow y'$ such that $\alpha_{ y \rightarrow y'} \neq 0$, $\supp(y)$ contains a species $A$ for which $\sgn(\sigma_A) = \sgn(\alpha_{ y \rightarrow y'})$,  where $\sigma_A$ denotes the term in sigma involving the species $A$ and $\sgn(\cdot)$ is the signum function.
    \item For each $y \rightarrow y'$ such that $\alpha_{ y \rightarrow y'} = 0$, either $\sigma_A = 0$ for all $A \in \supp(y)$, or else $\supp(y)$ contains species $A$ and $A'$ for which $\sgn(\sigma_A)=\sgn(\sigma_{A'})$, but not zero.
\end{enumerate}

A network that is not concordant is discordant.
    
\end{definition}

\begin{definition} \cite{SF2012}
    A kinetic system $(\mathscr{N},K)$ is \textbf{injective} if, for each pair of distinct stoicihometrically compatible vectors $x^*, x^{**} \in \mathbb{R}^\mathscr{S}_{\geq 0}$, at least one of which is positive,

$$\sum_{y \rightarrow y'} K_{y \rightarrow y'} (x^{**}) (y'-y) \neq \sum_{y \rightarrow y'} K_{y \rightarrow y'} (x^{*}) (y'-y).$$
\end{definition}
 
Note that an injective kinetic system is necessarily a monostationary system. Moreover, an injective kinetic system cannot admit two distinct stoichiometrically compatible equilibria, at least one of which is positive. 
 
\begin{definition}
    A kinetics $K$ for a reaction network $\mathscr{N}$ is \textbf{weakly monotonic} if, for each pair of vectors  $x^*, x^{**} \in \mathbb{R}^\mathscr{S}_{\geq 0}$, the following implications hold for each reaction $y \rightarrow y' \in \mathscr{R}$ such that $\supp(y) \subset \supp(x^*)$ and $\supp(y) \subset \supp(x^{**})$:
\begin{enumerate}
    \item $K_{y \rightarrow y'} (x^{**}) > K_{y \rightarrow y'} (x^{*}) \Longrightarrow$ there exists a species A $\in \supp(y)$ with $x_A^{**} > x_A^*$.
    \item $K_{y \rightarrow y'} (x^{**}) = K_{y \rightarrow y'} (x^{*}) \Longrightarrow x_A^{**} = x_A^*$ for all $A \in \supp(y)$ or else there are species $A,A' \in \supp(y)$ with  $x_A^{**} > x_A^*$ and $x_{A'}^{**} < x_{A'}^*$.
\end{enumerate}
    
\end{definition}

\subsection{Concentration robustness in kinetic systems}
\label{A2:robustness}
We begin with a brief review of the concept and some results on concentration robustness of species in kinetic systems. We are particularly interested in absolute and balanced concentration robustness, denoted by ACR and BCR respectively. \\

\textbf{Absolute concentration robustness} (ACR) pertains to a condition in which the concentration of a species in a kinetic system attains the same value in every positive steady-state set by parameters and does not depend on initial conditions. This concept was introduced by Shinar and Feinberg \cite{SF2012} published in Science in 2010. Notably, they presented sufficient structure-based conditions for a mass action system to display ACR on a particular species. This result was extended to power law systems with low deficiency \cite{FLRM2021,SF2010}, subsets of poly-PL kinetic systems \cite{LLMM2022}, and Hill-type kinetic systems \cite{HEME2021}. For larger systems and those with higher deficiency, independent decomposition helps identify ACR \cite{FEME2021}. In \cite{LMF2021}, a general approach extended the species hyperplane approach introduced in \cite{LLMM2022}. \\

\textbf{Balanced concentration robustness} (BCR) denotes the invariance of a species’ values on the subset of complex balanced equilibria of a system. Hence ACR $\Rightarrow$ BCR, but not conversely. \\

An important subset of power law kinetic systems consists of the log-parametrized (LP) system's, which consist of the positive equilibria log-parametrized (PLP) and complex balanced equilibria log parametrized systems (CLP) systems, which are defined as follows:

\begin{definition}
\label{def:CLP}
    A kinetic system $(\mathscr{N},K)$ is of type 
    \begin{enumerate}
        \item \textbf{PLP (positive equilibria log-parameterized)} if $E_+(\mathscr{N},K) \neq 0$ and $E_+(\mathscr{N},K) = \{ x \in \mathbb{R}^\mathscr{S}_{>0} | \log x - \log x^* \in (P_E)^\perp \}$, where $P_E$ is a subspace of $\mathbb{R}^\mathscr{S}$ and $x^*$ is a positive equilibrium.

        \item \textbf{CLP (complex-balanced equilibria log-parameterized)} if $Z_+(\mathscr{N},K) \neq 0$ and $Z_+(\mathscr{N},K) = \{ x \in \mathbb{R}^\mathscr{S}_{>0} | \log x - \log x^* \in (P_Z)^\perp \}$, where $P_Z$ is a subspace of $\mathbb{R}^\mathscr{S}$ and $x^*$ is a complex-balanced equilibrium.

        \item \textbf{bi-LP} if it is of PLP and of CLP type, and $P_E$ = $P_Z$.  
    \end{enumerate}
\end{definition}

In \cite{MURE2012}, it was shown that any complex balanced PL-RDK system is a CLP system with $P_Z = \tilde{S}^\perp$. It follows from the species hyperplane criterion for CLP systems, that a species has BCR if and only if its coordinate in all vectors of a basis is zero \cite{LLMM2022}. \\

\section{An illustration for the CF-WR algorithm}

We illustrate the algorithm using the MOLIB CRN realization of the biosynthesis of monolignol in \textit{Populus xylem} by Lee and Voit \cite{LEVO2010}. MOLIB is a weakly reversible NDK system. 
\begin{enumerate}
    \item[S1.] The NDK nodes are $X_2+X_{13}$,  $X_3+X_{13}$, $X_9+X_{13}$ and $X_{11}+X_{13}$. See Table \ref{MOLIB}. \\

        \begin{table}[h!]
        \centering
        \caption{CF-subsets of MOLIB}\label{MOLIB}%
        \begin{tabular}{llcc}
        \hline
        NF nodes & Reaction set & CF-subsets \\
        \hline
        $X_{2}+X_{13}$ & $R_{3},R_{4},R_{5}$ & $ \{R_{3}  \} ,\{R_{4}  \},\{R_{5}  \}$ \\
        $X_{3}+X_{13}$ & $R_{6},R_{7},R_{8}$ & $\{R_{6}\},\{R_{7},R_{8}  \} $ \\
        $X_{9}+X_{13}$ & $R_{16},R_{17}$ & $\{R_{16}  \},\{R_{17}  \}$ \\
        $X_{11}+X_{13}$ & $R_{19},R_{20}$ & $ \{R_{19} \},\{R_{20} \}$ \\

        \hline
        \hline
        \end{tabular}
        \end{table}

\item[S2.] Hence, we have
$$CF_{set} = \{ \{R_{3}  \} ,\{R_{4}  \},\{R_{5}  \}, \{R_{6}  \}, \{R_{7},R_{8}  \}, \{R_{16}  \},\{R_{17}  \}, \{R_{19} \},\{R_{20} \} \}.$$

    \begin{itemize}
        \item For the 1st iteration of S2, choose CF subset $\{R_{3}  \}$. Here, we compute

        $$\mathscr{N}_{CF_1}^{\text{cycle}} = \{R_{i} \text{ where } i \in \{1,2,3,6,9\} \}$$ 

        $$CF_{set} = \{\{R_{4}  \},\{R_{5}  \}, \{R_{7},R_{8}  \}, \{R_{16}  \},\{R_{17}  \}, \{R_{19} \},\{R_{20} \} \}.$$

        \item For the 2nd iteration of S2, choose CF subset $\{R_{4}  \}$. Here, we compute

        $$\mathscr{N}_{CF_2}^{\text{cycle}} = \{R_{i} \text{ where } i \in \{1,2,4,10,13,14,16 \} \}$$

        $$CF_{set} = \{\{R_{5}  \}, \{R_{7},R_{8}  \}, \{R_{17}  \}, \{R_{19} \},\{R_{20} \} \}.$$

        \item For the 3rd iteration of S2, choose CF subset $\{R_{5}  \}$. Here, we compute

        $$\mathscr{N}_{CF_3}^{\text{cycle}} = \{R_{i} \text{ where } i \in \{1,2,5 \} \}$$

        $$CF_{set} = \{ \{R_{7},R_{8}  \}, \{R_{17}  \}, \{R_{19} \},\{R_{20} \} \}.$$

        \item For the 4th iteration of S2, choose CF subset $\{R_{7},R_{8}  \}$. Here, we compute

        $$\mathscr{N}_{CF_4}^{\text{cycle}} = \{R_{i} \text{ where } i \in \{1,2,3,7,8,13,14,17,20,21\} \}.$$

        $$CF_{set} = \{ \{R_{19} \} \}.$$

        \item For the 5th iteration of S2, choose CF subset $\{ R_{19}  \}$. Here, we compute

        $$\mathscr{N}_{CF_5}^{\text{cycle}} = \{R_{i} \text{ where } i \in \{1,2,4,11,12,14,17,19\} \}.$$

        \item We update the $CF_{set} = \emptyset.$

\end{itemize}
\item[S3.] We compute for $\mathscr{R}_{set}$, the set of unassigned reactions.

$$\mathscr{R}_{set}=\{R_{15}, R_{18} \}$$

 \begin{itemize}
        \item For the 1st iteration of S3, choose the reaction $\{R_{15}  \}$. Here, we compute

        $$\mathscr{N}_{R_1}^{\text{cycle}} = \{R_{i} \text{ where } i \in \{1,2,4,11,15,18\} \}.$$ 

        \item We update $\mathscr{R}_{set}= \emptyset.$
\end{itemize}
\item[S4.] For brevity, we denote the  subnetworks $\mathscr{N}_{CF_1}^{cycle}$, $\mathscr{N}_{CF_2}^{cycle}$, $\mathscr{N}_{CF_3}^{cycle}$, $\mathscr{N}_{CF_4}^{cycle}$, $\mathscr{N}_{CF_5}^{cycle}$, $\mathscr{N}_{R_1}^{cycle}$ as $\mathscr{N}_1, \cdots, \mathscr{N}_{6}$, respectively. The occurrences of each reaction in the subnetworks are shown in Table \ref{tab:occur}. \\

        \begin{table}[h!]
        \centering
        \caption{Occurrences of $R_j$ in the $\mathscr{N}_i$}\label{tab:occur}%
        \begin{tabular}{lcccccccc}
        \hline
        Reaction & $\mathscr{N}_{1}$ & $\mathscr{N}_{2}$ & $\mathscr{N}_{3}$ & $\mathscr{N}_{4}$ & $\mathscr{N}_{5}$ & $\mathscr{N}_{6}$ & rate constant & kinetics $K_{R_j}^*$ \\
        \hline
       $R_1$ & Y & Y & Y & Y & Y & Y & $k_1/6$ & $\frac{k_1}{6} I_{K,1}(x)$  
       \\
       $R_2$ & Y & Y & Y & Y & Y & Y & $k_2/6$ & $\frac{k_2}{6} I_{K,2}(x)$  
       \\       
       $R_3$ & Y & Y & N & Y & N & N & $k_3/3$ & $\frac{k_3}{3} I_{K,3}(x)$  
       \\       
       $R_4$ & N & Y & N & N & Y & Y & $k_4/3$ & $\frac{k_4}{3} I_{K,4}(x)$  
       \\       
       $R_5$ & N & N & Y & N & N & N & $k_5$ & $k_5 I_{K,5}(x)$  
       \\       
       $R_6$ & Y & N & N & N & N & N & $k_6$ & $k_6 I_{K,6}(x)$  
       \\       
       $R_{7}$ & N & N & N & Y & N & N & $k_7$ & $k_7 I_{K,7}(x)$  
       \\       
       $R_{8}$ & N & N & N & Y & N & N & $k_8$ & $k_8 I_{K,8}(x)$  
       \\ \       
       $R_{9}$ & Y & N & N & N & N & N & $k_9$ & $k_9 I_{K,9}(x)$  
       \\       
       $R_{10}$ & N & Y & N & N & N & N & $k_{10}$ & $k_{10} I_{K,10}(x)$  
       \\       
       $R_{11}$ & N & N & N & N & Y & N & $k_{11}$ & $k_{11} I_{K,11}(x)$  
       \\     
       $R_{12}$ & N & N & N & N & Y & N & $k_{12}$ & $k_{12} I_{K,12}(x)$  
       \\      
       $R_{13}$ & N & Y & N & Y & N & N & $k_{13}/2$ & $\frac{k_{13}}{2} I_{K,13}(x)$  
       \\          
       $R_{14}$ & N & Y & N & Y & Y & N & $k_{14}/3$ & $\frac{k_{14}}{3} I_{K,14}(x)$  
       \\     
       $R_{15}$ & N & N & N & N & N & Y & $k_{15}$ & $k_{15} I_{K,15}(x)$  
       \\         
       $R_{16}$ & N & Y & N & N & N & N & $k_{16}$ & $k_{16} I_{K,16}(x)$  
       \\       
       $R_{17}$ & N & N & N & Y & Y & N & $k_{17}/2$ & $\frac{k_{17}}{2} I_{K,17}(x)$  
       \\        
       $R_{18}$ & N & N & N & N & N & Y & $k_{18}$ & $k_{18} I_{K,18}(x)$  
       \\        
       $R_{19}$ & N & N & N & N & Y & N & $k_{19}$ & $k_{19} I_{K,19}(x)$  
       \\          
       $R_{20}$ & N & N & N & Y & N & N & $k_{20}$ & $k_{20} I_{K,20}(x)$  
       \\        
       $R_{21}$ & N & N & N & Y & N & N & $k_{21}$ & $k_{21} I_{K,21}(x)$  
       \\ 
        \hline
        \hline
        \end{tabular}
        \end{table}

\item [S5.] The reactant set is
$\rho(\mathscr{R})= \{ X_{13}, X_i + X_{13} \text{ for } i \in \{1,2, \cdots 12 \} \}$. We select the complex $X_2$ for the following steps:
    \begin{itemize}
        \item For subnetwork $\mathscr{N}_1$, we translate each reaction $R_j$ using the operations dividing and shifting by adding $(0)X_2$. Thus, $\mathscr{N}^* = \mathscr{N}_1$. We have the following translated reactions

        $$\begin{array}{l}
            R_{3}^*: X_{2}+X_{13}  \rightarrow  X_{3}+X_{13} \\
            R_{6}^*: X_{3}+X_{13}  \rightarrow  X_{4}+X_{13} \\ 
            R_{9}^*: X_{4}+X_{13}  \rightarrow X_{13} \\
            R_{1}^*: X_{13} \rightarrow  X_{1}+X_{13} \\ 
            R_{2}^*: X_{1}+X_{13}  \rightarrow X_{2 }+X_{13} \\ 
            \end{array}$$ \\
           with kinetics $K_{R_j}^*$ for $R_j^*$ as shown in Table \ref{tab:occur}.

           \item For subnetwork $\mathscr{N}_2$, we translate each reaction $R_j$ using the operations dividing and shifting by adding $(1)X_2$. Thus, $\mathscr{N}^*_2$ have the following translated reactions

        $$\begin{array}{l}
            R_{4}^*: X_{2}+X_{13}+X_{1}  \rightarrow X_{5}+X_{13}+X_{1} \\
R_{10}^*: X_{5}+X_{13}+X_{2}  \rightarrow  X_{6}+X_{13}+X_{2} \\ 
R_{13}^*: X_{6}+X_{13}+X_{2}  \rightarrow  X_{8}+X_{13}+X_{2} \\ 
R_{14}^*: X_{8}+X_{13}+X_{2}  \rightarrow  X_{9}+X_{13}+X_{2} \\ 
R_{16}^*: X_{9}+ X_{13}+X_{2}  \rightarrow  X_{13}+X_{2} \\
R_{1}^*: X_{13}+X_{2}  \rightarrow  X_{1}+X_{13}+X_{2} \\ 
R_{2}^*: X_{1}+X_{13}+X_{2}  \rightarrow X_{2 }+X_{13}+X_{2} \\
R_{3}^*: X_{2}+X_{13}+X_{2}  \rightarrow  X_{3}+X_{13}+X_{2} \\
            \end{array}$$ \\
           with kinetics $K_{R_j}^*$ for $R_j^*$ as shown in Table \ref{tab:occur}.

           \item For subnetwork $\mathscr{N}_3$, we translate each reaction $R_j$ using the operations dividing and shifting by adding $(2)X_2$. Thus, $\mathscr{N}^*_3$ have the following translated reactions

        $$\begin{array}{l}
            R_{1}^*: X_{13}+2X_{2} \rightarrow  X_{1}+X_{13}+2X_{2} \\ 
R_{2}^*: X_{1}+X_{13}+2X_{2}  \rightarrow X_{2 }+X_{13}+2X_{2} \\
R_{5}^*: X_{2}+X_{13}+2X_{2}  \rightarrow  X_{13}+2X_{2} \\ 
            \end{array}$$ \\
           with kinetics $K_{R_j}^*$ for $R_j^*$ as shown in Table \ref{tab:occur}.

           \item For subnetwork $\mathscr{N}_4$, we translate each reaction $R_j$ using the operations dividing and shifting by adding $(3)X_2$. Thus, $\mathscr{N}^*_4$ have the following translated reactions

        $$\begin{array}{l}
            R_{7}^*: X_{3}+X_{13}+3X_{2}  \rightarrow X_{6 }+X_{13}+3X_{2} \\   
R_{8}^*: X_{3}+X_{13}+3X_{2}  \rightarrow X_{13}+3X_{2} \\    
R_{13}^*: X_{6}+X_{13}+3X_{2}  \rightarrow  X_{8}+X_{13}+3X_{2} \\    
R_{14}^*: X_{8}+X_{13}+3X_{2}  \rightarrow  X_{9}+X_{13}+3X_{2} \\    
R_{17}^*: X_{9}+ X_{13}+3X_{2}  \rightarrow X_{11} + X_{13} + 3X_{2} \\   
R_{20}^*: X_{11}+X_{13}+3X_{2}  \rightarrow X_{12}+X_{13}+3X_{2} \\   
R_{21}^*: X_{12}+X_{13}+3X_{2}  \rightarrow  X_{13}+3X_{2} \\    
R_{1}^*: X_{13}+3X_{2}  \rightarrow  X_{1}+X_{13}+3X_{2} \\    
R_{2}^*: X_{1}+X_{13}+3X_{2}  \rightarrow X_{2 }+X_{13}+3X_{2} \\ 
R_{3}^*: X_{2}+X_{13}+3X_{2}  \rightarrow  X_{3}+X_{13}+3X_{2} \\  
            \end{array}$$ \\
           with kinetics $K_{R_j}^*$ for $R_j^*$ as shown in Table \ref{tab:occur}.

           \item For subnetwork $\mathscr{N}_5$, we translate each reaction $R_j$ using the operations dividing and shifting by adding $(4)X_2$. Thus, $\mathscr{N}^*_5$ have the following translated reactions

        $$\begin{array}{l}
            R_{19}^*: X_{11}+X_{13}+4X_{2}  \rightarrow  X_{13}+4X_{2} \\   
   
R_{1}^*: X_{13}+4X_{2}  \rightarrow  X_{1}+X_{13}+4X_{2} \\   
  
R_{2}^*: X_{1}+X_{13}+4X_{2}  \rightarrow X_{2 }+X_{13}+4X_{2} \\   
  
R_{4}^*: X_{2}+X_{13}+4X_{2}  \rightarrow X_{5}+X_{13}+4X_{2} \\   
   
R_{11}^*: X_{5}+X_{13}+4X_{2}  \rightarrow  X_{7}+ X_{13}+4X_{2} \\   
   
R_{12}^*: X_{7}+X_{13}+4X_{2}  \rightarrow  X_{8 }+X_{13}+4X_{2} \\   
  
R_{14}^*: X_{8}+X_{13}+4X_{2}  \rightarrow  X_{9}+X_{13}+4X_{2} \\   
   
R_{17}^*: X_{9}+ X_{13}+4X_{2}  \rightarrow X_{11}+ X_{13}+4X_{2} \\ 
            \end{array}$$ \\
           with kinetics $K_{R_j}^*$ for $R_j^*$ as shown in Table \ref{tab:occur}.

           \item For subnetwork $\mathscr{N}_6$, we translate each reaction $R_j$ using the operations dividing and shifting by adding $(5)X_2$. Thus, $\mathscr{N}^*_6$ have the following translated reactions

        $$\begin{array}{l}
            R_{15}^*: X_{7}+X_{13}+5X_{2}  \rightarrow  X_{10}+X_{13}+5X_{2} \\   

R_{18}^*: X_{10}+X_{13}+5X_{2} \rightarrow  X_{13}+5X_{2} \\   
 
R_{1}^*: X_{13}+5X_{2} \rightarrow  X_{1}+X_{13}+5X_{2} \\   
 
R_{2}^*: X_{1}+X_{13}+5X_{2}  \rightarrow X_{2 }+X_{13}+5X_{2} \\   

R_{4}^*: X_{2}+X_{13}+5X_{2}  \rightarrow X_{5}+X_{13}+5X_{2} \\   
 
R_{11}^*: X_{5}+X_{13}+5X_{2} \rightarrow  X_{7}+ X_{13}+5X_{2} \\ 
            \end{array}$$ \\
           with kinetics $K_{R_j}^*$ for $R_j^*$ as shown in Table \ref{tab:occur}.
\item Lastly, we have $\mathscr{N}^* = \cup \mathscr{N}^*_i$ with corresponding kinetics $K^*$. We call this realization, MOLIB*.
    \end{itemize}

\end{enumerate}

Table \ref{netprop:MOLIB*} shows the difference in some basic network properties of the MOLIB* and MOLIB. The CRNToolbox Concordance Report shows that MOLIB* remains concordant despite the use of dividing in the algorithm.
\begin{table}[h!]
\centering
\caption{Network numbers and properties of MOLIB* and MOLIB}\label{netprop:MOLIB*}%
\begin{tabular}{lcc}
\hline
Network numbers & MOLIB* & MOLIB \\
\hline
Number of species ($m$) & 13 & 13 \\
Number of complexes ($n$) & 36 & 13 \\
Number of irreversible reactions ($rirr$) & 36 & 21 \\
Number of reversible reactions ($rrev$) & 0 & 0 \\
Number of reactions ($r$) & 36 & 21 \\
Number of linkage classes ($l$) & 6 & 1 \\
Number of strong linkage classes ($sl$) & 6 & 1 \\
Number of terminal strong linkage classes ($t$) & 6 & 1 \\
Rank of network ($s$) & 12 & 12 \\
Deficiency ($\delta$) & 18 & 0 \\
\hline
\hline
\end{tabular}
\end{table}

\section{Table of Symbol and Glossary of Acronyms}
We list some of the symbols and acronyms used in the paper.

\begin{table}[ht!]
\caption{List of Symbols}
\begin{tabular}{ll}
$\mathscr{S}$ & set of species \\
$m$ & cardinality of species \\
$\mathscr{C}$ & set of complexes \\
$n$ & cardinality of complexes \\
$\mathscr{R}$ & set of reactions \\
$r$ & cardinality of reactions \\
$l$ & number of linkage classes \\
$t$ & number of terminal strong linkage classes \\
$\delta$	& 	Deficiency	\\
$S$ & Stoichiometric subspace \\
$S_+$ & Stoichiometric cone \\
$\mathscr{N}$ & Reaction network \\
$N$ & Stoichiometric matrix \\
$K$	& Kinetics of a CRN	\\
$I_K$	& 	Interaction map	\\
$I_a$	& 	Incidence mmatrix	\\
$F$	& 	Kinetic order matrix	\\
$Y$	& 	Matrix of complexes	\\
$E_+(\mathscr{N},K)$	& 	Set of positive equilibria	\\
$Z_+(\mathscr{N},K)$	& 	Set of complex balanced equilibria	\\
$N_R$ & number of CF subsets \\
$n_R$ & number of reactant complexes \\
$r_{mcf}$ & number of reactions in the maximal CF-subsystems \\
\end{tabular}
\end{table}

\begin{table}[ht!]
\caption{Abbreviations}
\begin{tabular}{ll}
ACB & Absolutely complex balanced \\
ACR	&	Absolute concentration robustness	\\
BCR	&	Balanced concentration robustness	\\
CFK	&	Complex factorizable kinetics	\\
CLP & complex balanced equilibria log parametrized systems \\
FSK & Factor Span Surejective kinetics \\
HLK & Hill-type kinetics \\
ILC & independent linkage classes \\
ISS & Interaction Span Surjective \\
LP & Log Parametrization \\
MAK	&	Mass action kinetics	\\
NFK & Non-complex factorizable kinetics \\
NIK & Non-inhibitory kinetics \\
NDK 	&	Non-reactant determined kinetics	\\
ODE & Ordinary differential equations \\
PLP & positive equilibria log-parametrized \\
PYK	&	Poly-PL kinetics	\\
PL	&	Power law	\\
RDK	&	Reactant-determined kinetics	\\
RIDK & rate constant-interaction map decomposable kinetics \\
RN	&	Reaction network	\\
rsp & reactant support preserving \\
SFRF & Species formation rate functions \\
\end{tabular}
\end{table}

\end{document}